\documentclass[12pt]{article}
\usepackage{amsmath}
\usepackage{amsfonts}
\usepackage{amssymb}
\usepackage{mathrsfs}
\usepackage{amsthm}
\usepackage{mathtools}
\usepackage{graphicx}
\usepackage{enumerate}
\usepackage{epstopdf}
\usepackage{sectsty}
\usepackage{chngcntr}
\usepackage{bbm}
\usepackage[round]{natbib}
\usepackage[toc,page]{appendix}
\usepackage{multirow}
\usepackage{authblk}
\usepackage{setspace}
\usepackage{url}

\newcommand\numberthis{\addtocounter{equation}{1}\tag{\theequation}}
\newcommand{\mbb}[1]{\mathbb{#1}}
\newcommand{\mbf}[1]{\boldsymbol{#1}}
\newcommand{\mcal}[1]{\mathcal{#1}}
\newcommand{\mrm}[1]{\textrm{#1}}

\theoremstyle{definition}
\newtheorem{theorem}{Theorem}
\newtheorem{remark}{Remark}
\newtheorem{mydef}{Definition}
\newtheorem{lemma}{Lemma}
\newtheorem{exmp}{Example}
\newtheorem{prop}[theorem]{Proposition}
\newtheorem{corollary}{Corollary}

\makeatletter
\newcommand\incircbin
{%
  \mathpalette\@incircbin
}
\newcommand\@incircbin[2]
{%
  \mathbin%
  {%
    \ooalign{\hidewidth$#1#2$\hidewidth\crcr$#1\bigcirc$}%
  }%
}
\newcommand{\owedge}{\incircbin{\wedge}}
\makeatother

\begin{document}

\def\spacingset#1{\renewcommand{\baselinestretch}%
{#1}\small\normalsize} \spacingset{1}


  \title{\bf A Kernel-Based Neural Network for High-dimensional Genetic Risk Prediction Analysis}
  \author{Xiaoxi Shen, 
    Xiaoran Tong
    and \\
    Qing Lu}
	\date{}
  \maketitle

\bigskip
\begin{abstract}
Risk prediction capitalizing on emerging human genome findings holds great promise for new prediction and prevention strategies. While the large amounts of genetic data generated from high-throughput technologies offer us a unique opportunity to study a deep catalog of genetic variants for risk prediction, the high-dimensionality of genetic data and complex relationships between genetic variants and disease outcomes bring tremendous challenges to risk prediction analysis. To address these rising challenges, we propose a kernel-based neural network (KNN) method. KNN inherits features from both linear mixed models (LMM) and classical neural networks and is designed for high-dimensional risk prediction analysis. To deal with datasets with millions of variants, KNN summarizes genetic data into kernel matrices and use the kernel matrices as inputs. Based on the kernel matrices, KNN builds a single-layer feedforward neural network, which makes it feasible to consider complex relationships between genetic variants and disease outcomes. The parameter estimation in KNN is based on MINQUE and we show, that under certain conditions, the average prediction error of KNN can be smaller than that of LMM. Simulation studies also confirm the results.
\end{abstract}

\noindent%
\vfill

\newpage
\section{Introduction}
\label{sec:intro}

Linear mixed effect models are powerful tools to model complex data structures. By adding random effects into the model, it becomes feasible to model correlated observations. Moreover, it is also possible in linear mixed effect models to make best predictions on the random effects. In genetic studies, more advantages of linear mixed effect models have been explored. For instance, in genome-wide association studies (GWAS) \citep{bush2012genome}, a simple linear regression is conducted on each single-nucleotide polymorphism (SNP) so that there are a large number of hypothesis to be tested and we also need to deal with the multiple test correction issue as well. On the other hand, if we think the genetic effect as a random effect, the null hypothesis we are going to test reduces to testing whether the variance component of the random effect is zero or not. Here are some well-known applications: in sequence kernel association test (SKAT) proposed by \citet{wu2011rare}, a score type test based on a mixed effect model is used to perform the genetic effect; \citet{yang2011gcta} created the tool known as the genome-wide complex trait analysis (GCTA), which is also based on linear mixed model, to address the ``missing heretability" problem.

In this paper, we propose a method, which we call it the kernel neural network (KNN) for high-dimensional risk prediction analysis. As will be seen in the paper, under certain scenarios, our model can reduce to a linear mixed effect model. We call such method KNN in that it also inherits an important property from the neural network, which is that the model can be used to consider nonlinear effects. Due to the complex structure of such method, it is difficult to obtain estimators for the parameters in the model. To make things worse, not all the parameters are identifiable. To address such issue, instead of using the popular likelihood type inference using the restricted maximum likelihood estimator (REML) \citep{corbeil1976restricted}, we use the minimum quadratic unbiased estimator (MINQUE) proposed by \citet{rao1970estimation, rao1971estimation, rao1972estimation} to estimate the ``variance components". We show both theoretically and empirically that the model has some interesting properties.

The remaining paper is arranged as follows: Section 2 provides the basic description of the KNN and the estimation procedure for the parameters; Section 3 provides some discussion on how to make predictions using KNN and followed by some simulation results in Section 4. Before we proceeding to the main text, we first summarize some notations that will be frequently used in this paper.

\textit{Notations}: Throughout the paper, capital bold italic letters $\mbf{A},\mbf{B},\ldots, \mbf{\Gamma},\mbf{\Theta},\ldots$ will be used to denote matrices; small bold italic letters $\mbf{a},\mbf{b},\ldots,\mbf{\alpha},\mbf{\beta},\ldots$ will be used to denote vectors and other small letters will be used to denote scalars. $\mbf{I}_n$ will be used to denote an $n\times n$ identity matrix and the symbol ``$\precsim$" will be used to denote asymptotically less than.

\section{Methodologies}
\label{sec:meth}
Kernel methods are widely used in recent machine learning due to its capability of capturing nonlinear features from the data so that the prediction error can be diminished. As has been mentioned in \citet{shawe2004kernel}, given a kernel and a training set, the kernel matrix acts as an information bottleneck, as all the information available to a kernel algorithm must be extracted from this matrix. On the other hand, linear mixed effect models are also widely used in the area of genetic risk prediction \citep{yang2011gcta}. Therefore, it seems natural to combine these two methods together. A naive way is two change the covariance matrix of the random effect in the linear mixed model to a kernel matrix. For instance, in \citet{yang2011gcta}, they consider the following linear mixed model:
\begin{equation}\label{Eq: GCTAModel}
	\mbf{y}=\mbf{Z\beta}+\mbf{a}+\mbf{\epsilon},
\end{equation}
where $\mbf{y}\in\mbb{R}^n$ is a vector of phenotypes; $\mbf{Z}$ is the design matrix for fixed effects $\mbf{\beta}$; $\mbf{a}\in\mbb{R}^n$ is the total genetic effects of the individuals with $\mbf{a}\sim\mcal{N}_n\left(\mbf{0},\sigma_a^2\mbf{K}\right)$, and $\mbf{K}$ can be interpreted as the genetic relationship matrix between subjects; $\mbf{\epsilon}\sim\mcal{N}_n(\mbf{0},\sigma_\epsilon^2\mbf{I}_n)$. Such model can also be written into the following hierarchical structure:
\begin{align*}
	\mbf{y}|\mbf{Z},\mbf{a} & \sim\mcal{N}_n\left(\mbf{Z\beta}+\mbf{a},\sigma_\epsilon^2\mbf{I}_n\right)\\
	\mbf{a} & \sim\mcal{N}_n\left(\mbf{0},\sigma_a^2\mbf{K}\right).
\end{align*}
To model more complex genotype-to-phenotype relationships, what we did is a little step further, we start by creating some latent features from the kernel matrix. Based on these latent features, higher order kernel matrices can be constructed, which will be used as the covariance matrix for the random effect. 

We now explain the model in more details. Consider that the phenotype $\mbf{y}$ is modeled as a random effect model: given some latent variables $\mbf{u}_1,\ldots,\mbf{u}_m$,
\begin{align*}
	\mbf{y}|\mbf{Z},\mbf{a} & \sim\mcal{N}_n\left(\mbf{Z\beta}+\mbf{a},\phi\mbf{I}_n\right)\\
	\mbf{a}|\mbf{u}_1,\ldots,\mbf{u}_m & \sim\mcal{N}_n\left(\mbf{0},\sum_{j=1}^J\tau_j\mbf{K}_j(\mbf{U})\right),
\end{align*}
that is the covariance matrix of the random effect $\mbf{a}$ is a positive combination of some latent kernel matrices $\mbf{K}_j(\mbf{U})$ constructed based on latent variables $\mbf{u}_1,\ldots,\mbf{u}_m$ and we let $\mbf{U}=[\mbf{u}_1,\cdots\mbf{u}_m]\in\mbb{R}^{n\times m}$. Moreover, the latent variable $\mbf{u}_i$ is modeled as follow:
$$
\mbf{u}_1,\ldots,\mbf{u}_m\sim\mrm{ i.i.d.  }\mcal{N}_n\left(\mbf{0},\sum_{l=1}^L\xi_l\mbf{K}_l(\mbf{X})\right),
$$
where in the model, $n$ is the sample size; $m$ is the number of hidden units in the network; $\mbf{K}_l(\mbf{X}), l=1,\ldots,L$ are kernel matrices constructed based on the genetic variables. For instance, if we have $p$ genetic variants ($p$ can be greater than $n$), we can define $\mbf{K}(\mbf{X})=p^{-1}\mbf{XX}^T$, which is known as the product kernel.

As an illustration, the basic hierarchical structure of the model can be seen from Figure \ref{KNNfig}. Due to the similarity in the network structure as in the case of popular neural network, we thus call our model a kernel neural network (KNN). KNN has several nice features. First, by considering genetic effects as random effects, the method can simultaneously deal with millions of variants. This addresses the limitation of the fixed-effect conventional neural network, which is computationally prohibitive on such a large number of variables. Second, by using random genetic effects, the model complexity is also greatly reduced, as we no longer require to estimate a large number of fixed genetic effects. Third, KNN allows for a large number of hidden units without increasing model complexity. Finally, using hidden units, KNN can capture non-linear and non-additive effects, and therefore is able to model complex functions beyond linear model.

\begin{figure}[htbp]
	\centering
	\includegraphics[width=\textwidth]{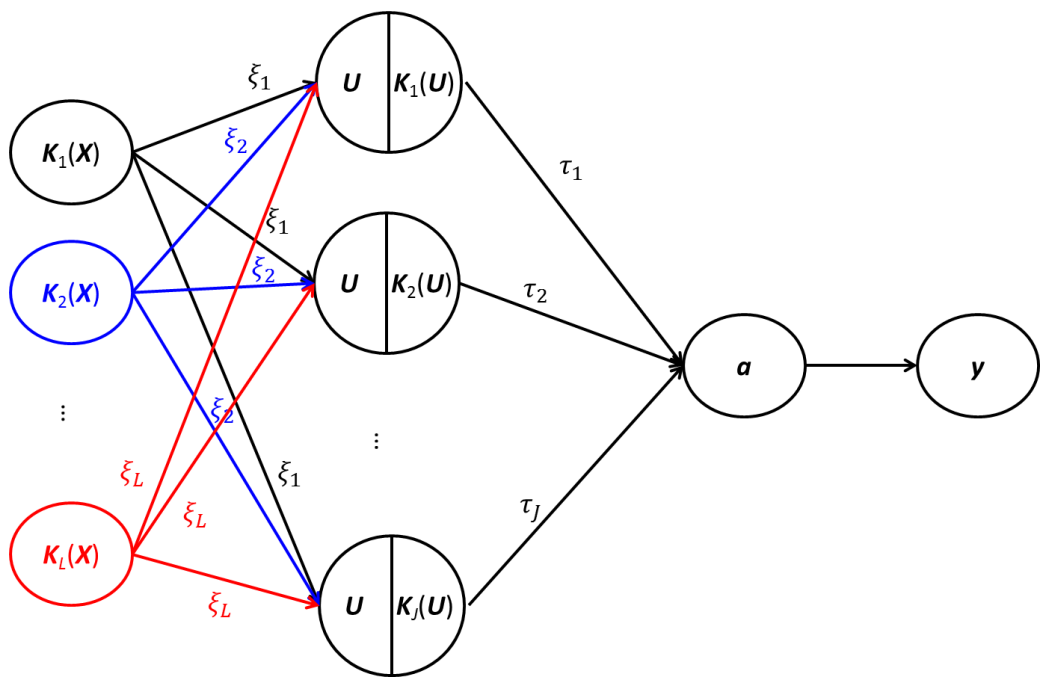}
	\caption{An illustration of the hierarchical structure of the kernel neural network model.}\label{KNNfig}
\end{figure}

In the remaining part of this section and section 3, we will focus on the scenario where there is no fixed effects, that is $\mbf{\beta}=\mbf{0}$. In section 4, we will consider the general estimation procedure when $\mbf{\beta}\neq\mbf{0}$ and we will also calculate the prediction error. 

\subsection{Quadratic Estimators for Variance Components}
Popular estimation strategies for variance components in linear models are the maximum likelihood estimator (MLE) and the restricted maximum likelihood estimator (REML) \citep{corbeil1976restricted}. However, both methods depend on the marginal distribution of $\mbf{y}$. In our kernel neural network (KNN) model, it is generally difficult to obtain the marginal distribution of $\mbf{y}$, which involves high dimensional integration with respect to $\mbf{u}_1,\ldots,\mbf{u}_m$. Moreover, the $\mbf{u}_i$'s are embedded in the kernel matrix $\mbf{K}(\mbf{U})$, which makes the integration even more complicated.

On the other hand, given the model described in the previous paragraph, we can easily know the conditional distribution of $\mbf{y}|\mbf{u}_1,\ldots,\mbf{u}_m$:
$$
\mbf{y}|\mbf{u}_1,\ldots,\mbf{u}_m\sim\mcal{N}_n\left(\mbf{0},\sum_{j=1}^J\tau_j\mbf{K}_j(\mbf{U})+\phi\mbf{I}_n\right).
$$
Then the marginal mean and variance of $\mbf{y}$ can be obtained via conditioning arguments:
\begin{align*}
	\mbb{E}[\mbf{y}] & =\mbb{E}\left(\mbb{E}[\mbf{y}|\mbf{u}_1,\ldots,\mbf{u}_m]\right)=\mbf{0}.\\
	\mrm{Var}[\mbf{y}] & =\mbb{E}\left[\mrm{Var}\left(\mbf{y}|\mbf{u}_1,\ldots,\mbf{u}_m\right)\right]+\mrm{Var}\left[\mbb{E}\left(\mbf{y}|\mbf{u}_1,\ldots,\mbf{u}_m\right)\right]\\
	& =\mbb{E}\left[\sum_{j=1}^J\tau_j\mbf{K}_j(\mbf{U})+\phi\mbf{I}_n\right]\\
	& =\sum_{j=1}^J\tau_j\mbb{E}[\mbf{K}_j(\mbf{U})]+\phi\mbf{I}_n.\\
	& :=\sum_{j=0}^J\tau_j\mbb{E}\left[\mbf{K}_j(\mbf{U})\right],
\end{align*}
where $\tau_0=\phi$ and $\mbf{K}_0(\mbf{U})=\mbf{I}_n$. Given the marginal mean and covariance matrix, the MInimum Quadratic Unbiased Estimator (MINQUE) proposed by \citet{rao1970estimation, rao1971estimation, rao1972estimation} can be used to estimate the variance components. The basic idea of MINQUE is to use a quadratic form $\mbf{y}^T\mbf{\Theta}\mbf{y}$ to estimate a linear combination of variance components. The MINQUE matrix $\mbf{\Theta}$ is obtained by minimizing a suitable matrix norm, which is typically chosen to be the Frobenius norm, of the difference between $\mbf{\Theta}$ and the matrix in the quadratic estimator by assuming that we know the random components in the linear models. The constraint in the optimization problem is to guarantee the unbiasedness of the estimators. One advantage of MINQUE is that it has a closed form solution provided by Lemma 3.4 in \citet{rao1971estimation} so that it can be computed efficiently.
However, MINQUE can also provide a negative estimate for a single variance component. When this occurs, we simply set the negative estimators to be zero, as we usually did for MLE and REML of variance components, except for the error variance component. If the MINQUE estimate for error variance component becomes negative, we project the MINQUE matrix $\mbf{\Theta}$ onto the positive semi-definite cone $\mcal{S}_n^+$. Specifically, the modified estimate for error component becomes
$$
\hat{\tau}_0=\mbf{y}^T\mbf{O}\begin{bmatrix}
	\max\{\lambda_1,0\} & & \\
	& \ddots & \\
	& & \max\{\lambda_n,0\}
\end{bmatrix}\mbf{O}^T\mbf{y},
$$
where $\mbf{O}\mrm{diag}\{\lambda_1,\ldots,\lambda_n\}\mbf{O}^T$ is the eigen-decomposition of $\mbf{\Theta}$.

\subsection{MINQUE in KNN}
For the ease of theoretical justifications, throughout the remaining of the paper, we illustrate the method with one kernel matrix with the form of $\mbf{K}_{ij}(\mbf{U})=f\left[\frac{1}{m}\mbf{w}_i^T\mbf{w}_j\right]$, where $\mbf{w}_1,\ldots,\mbf{w}_n$ are the rows of the matrix $\mbf{U}$.

We start by considering the simplest case $f(x)=x$, in which case, the kernel matrix becomes $\mbf{K}(\mbf{U})=\frac{1}{m}\mbf{U}\mbf{U}^T$. In this case, we have an explicit form of the marginal variance of $\mbf{y}$. Since $\mbf{U}\mbf{U}^T\sim\mcal{W}_n(m, \sum_{l=1}^L\xi_l\mbf{K}_l(\mbf{X}))$, we have
\begin{equation}\label{marginalVaryUnid}
	\mbf{V}:=\mrm{Var}[\mbf{y}]=\tau\mbb{E}[\mbf{K}(\mbf{U})]+\phi\mbf{I}_n=\sum_{l=1}^L\tau\xi_l\mbf{K}_l(\mbf{X})+\phi\mbf{I}_n.
\end{equation}
From equation (\ref{marginalVaryUnid}), we can see that there is an identifiability issue if we directly estimate $\tau$ and $\xi_l$. To resolve this issue, we reparameterize the covariance matrix by letting $\theta_l=\tau\xi_l$, $\theta_0=\phi$, $\mbf{K}_0(\mbf{X})=\mbf{I}_n$ and rewrite $\mrm{Var}[\mbf{y}]$ as
\begin{equation}\label{marginalVaryId}
	\mbf{V}=\sum_{l=0}^L\theta_l\mbf{K}_l(\mbf{X})=\sum_{l=0}^L\theta_l\mbf{S}_l(\mbf{X})\mbf{S}_l^T(\mbf{X}),
\end{equation}
where $\mbf{S}_0(\mbf{X}),\ldots,\mbf{S}_L(\mbf{X})$ are the Cholesky lower triangles for the kernel matrices $\mbf{K}_0(\mbf{X}),\ldots,\mbf{K}_L(\mbf{X})$, respectively. Then the parameters $\theta_0,\theta_1,\ldots,\theta_L$ can be estimated via MINQUE. Specifically
\begin{align*}
	\hat{\theta}_0 & =\mbf{y}^T\hat{\mbf{A}}_0\mbf{y}\\
	\hat{\theta}_i & =\mbf{y}^T\mbf{A}_i\mbf{y}\vee0,\quad i = 1,\ldots,L,
\end{align*}
where
$$
\mbf{A}_i=\sum_{l=0}^{L}\eta_l\left[\sum_{l=0}^L\mbf{K}_l(\mbf{X})\right]^{-1}\mbf{K}_l(\mbf{X})\left[\sum_{l=0}^L\mbf{K}_l(\mbf{X})\right]^{-1},\quad i=0,\ldots,L
$$
and $\eta_0,\ldots,\eta_L$ are the solutions to $\mbf{\Gamma\eta}=\mbf{e}_i$, where $\mbf{e}_i$ is a vector of zero except that the $i$th element is 1 and 
$$
\mbf{\Gamma}_{ij}=\mrm{tr}\left(\left[\sum_{l=0}^L\mbf{K}_l(\mbf{X})\right]^{-1}\mbf{K}_i(\mbf{X})\left[\sum_{l=0}^L\mbf{K}_l(\mbf{X})\right]^{-1}\mbf{K}_j(\mbf{X})\right).
$$
Moreover, $\hat{\mbf{A}}_0=P_{\mcal{S}_n^+}\mbf{A}_0$ as mentioned above. For general kernel matrix of the form $\mbf{K}(\mbf{U})=f\left[\frac{1}{m}\mbf{U}\mbf{U}^T\right]$, where $f[\mbf{B}]$ means that we apply the map $f:\mbb{R}\to\mbb{R}$ elementwisely to the matrix $\mbf{B}$ and $f$ is a function on $\mbb{R}$ satisfying the following property, which we called the Generalized Linear Separable Condition: 
\begin{equation}
	f\left(\sum_{\alpha=1}^kc_\alpha x_\alpha\right)=\sum_{\alpha=1}^{k'}g_\alpha(c_1,\ldots,c_k)h_\alpha(x_1,\ldots,x_k),
\end{equation}
where $c_1,\ldots,c_k\in\mbb{R}$ are coefficients and $g_1,\ldots,g_k, h_1,\ldots,h_k$ are some functions. Examples of kernel functions satisfying the condition are polynomial kernels. We know that
\begin{equation}\label{KernelOfFuncProd}
	\mbf{K}_{st}(\mbf{U})=f\left(\frac{\mbf{w}_s^T\mbf{w}_t}{m}\right),\quad s,t=1,\ldots,n
\end{equation}
and 
\begin{equation}\label{distOfWiWj}
	\begin{bmatrix}
		\mbf{w}_{i1}\\
		\mbf{w}_{j1}
	\end{bmatrix},\ldots,\begin{bmatrix}
		\mbf{w}_{im}\\
		\mbf{w}_{jm}
	\end{bmatrix}\sim\mrm{i.i.d. }\mcal{N}_2\left(\begin{bmatrix}
		0\\
		0
	\end{bmatrix},\begin{bmatrix}
		\sigma_{ii} & \sigma_{ij}\\
		\sigma_{ij} & \sigma_{jj}
	\end{bmatrix}
	\right),
\end{equation}
where $\sigma_{ii},\sigma_{jj}, \sigma_{ij}$ are the corresponding elements in $\mbf{\Sigma}=\sum_{l=1}^L\xi_l\mbf{K}_l(\mbf{X})$. We can then use Taylor expansion to obtain an approximation of $\mbb{E}[\mbf{K}_{ij}(\mbf{U})]$, which is given in Lemma \ref{funcProdApproxLm}. The proof of Lemma \ref{funcProdApproxLm} can be found in Appendix B.

\begin{lemma}\label{funcProdApproxLm}
	Let the random vectors $\mbf{w}_i, \mbf{w}_j\in\mbb{R}^m$ be as in (\ref{distOfWiWj}) and the $\mbf{K}_{ij}(\mbf{U})$ as defined in (\ref{KernelOfFuncProd}). Then if 
	\begin{equation}\label{asBoundedAssumption}
		\left|f''\left(\lambda\sigma_{ij}+(1-\lambda)\frac{\mbf{w}_i^T\mbf{w}_j}{m}\right)\right|\leq M,\quad\mrm{a.s.},
	\end{equation}
	for some $M>0$ and all $\lambda\in[0,1]$, we have
	\begin{align*}
		\mbb{P}\left(\left|\mbf{K}_{ij}(\mbf{U})-\hat{\mbf{K}}_{ij}(\mbf{U})\right|>\delta\right)
		& \leq4\exp\left\{-m\left(1\wedge\frac{\delta}{20Ms_{ij}}\wedge\frac{\delta}{M\sigma_{ij}^2}\wedge\frac{
			1}{4|\sigma_{ij}|}\sqrt{\frac{2\delta}{M}}\right)\right\},
	\end{align*}
	where $\hat{\mbf{K}}_{ij}(\mbf{U})=f(\sigma_{ij})+f'(\sigma_{ij})\left(\frac{\mbf{w}_i^T\mbf{w}_j}{m}-\sigma_{ij}\right).$
\end{lemma}

Lemma \ref{funcProdApproxLm} and Remark \ref{funcProdApproxRmk} in Appendix B show that when we use $\hat{K}(\mbf{U})=[\hat{\mbf{K}}_{ij}(\mbf{U})]$ to approximate $\mbf{K}(\mbf{U})$ and the number of hidden features $\mbf{u}_1,\ldots,\mbf{u}_m$ is large enough, the approximation will be sufficiently small. Hence, we can write $\mbf{K}(\mbf{U})$ as follow:

\begin{equation}\label{KmatApprox}
	\mbf{K}(\mbf{U})=\hat{\mbf{K}}(\mbf{U})+o_P(1)=f[\mbf{\Sigma}]+f'[\Sigma]\odot\left(\frac{1}{m}\mbf{U}\mbf{U}^T-\mbf{\Sigma}\right)+o_P(1),
\end{equation}
where $\odot$ means the Hadamard product of two matrices. Moreover, the Strong Law of Large Numbers implies $\frac{1}{m}\mbf{U}\mbf{U}^T\to\mbf{\Sigma}$ a.s. Therefore, equation (\ref{KmatApprox}) can be further written as
$$
\mbf{K}(\mbf{U})=f[\mbf{\Sigma}]+o_P(1),
$$
i.e., $\mbf{K}(\mbf{U})\xrightarrow{P}f[\mbf{\Sigma}]$ as $m\to\infty$ element-wisely. Following from the version of Dominated Convergence Theorem based on convergence in probability, the following lemma can be easily shown.

\begin{lemma}\label{EFuncProdApproxLm}
	Under the assumptions of Lemma \ref{funcProdApproxLm}, if $f''\left(\eta_{ij}\right)\left(\frac{\mbf{w}_i^T\mbf{w}_j}{m}-\sigma_{ij}\right)^2\in L^1(\mbb{P})$, then
	$$
	\mbb{E}\left[\frac{1}{2}f''\left(\eta_{ij}\right)\left(\frac{\mbf{w}_i^T\mbf{w}_j}{m}-\sigma_{ij}\right)^2\right]=o(1),
	$$
	where $\eta_{ij}=\lambda\sigma_{ij}+(1-\lambda)\frac{\mbf{w}_i^T\mbf{w}_j}{m}$ for some $\lambda\in[0,1]$.
\end{lemma}

Based on Lemma \ref{EFuncProdApproxLm}, the marginal variance-covariance matrix $\mbf{V}$ can be written as
\begin{align*}
	\mbf{V} & =\tau\mbb{E}\left[\mbf{K}(\mbf{U})\right]+\phi\mbf{I}_n\\
	& \simeq\tau\mbb{E}\left[\hat{\mbf{K}}(\mbf{U})\right]+\phi\mbf{I}_n\\
	& \simeq\tau f[\mbf{\Sigma}]+\phi\mbf{I}_n\\
	& =\tau\sum_{l=1}^{L'}g_l(\xi_1,\ldots,\xi_L)h_l\left[\mbf{K}_1(\mbf{X}),\ldots,\mbf{K}_L(\mbf{X})\right]+\phi\mbf{I}_n\\
	& =\sum_{l=0}^{L'}\theta_l\mbf{S}_l(\mbf{X})\mbf{S}_l^T(\mbf{X}),
\end{align*}
where $\theta_0=\phi$, $\theta_l=\tau g_l(\xi_1,\ldots,\xi_L)$, $l=1,\ldots,L'$ and $S_0(\mbf{X})=\mbf{I}_n$, $\mbf{S}_l(\mbf{X})$ is the Cholesky lower triangle for the matrix $h_l\left[\mbf{K}_1(\mbf{X}),\ldots,\mbf{K}_L(\mbf{X})\right]$, $l=1,\ldots,L$. As an example, we may consider $f(x)=(1+x)^2$, which corresponds to the output polynomial kernel. $f[\mbf{\Sigma}]=(\mbf{J}+\xi_1\frac{1}{p}\mbf{XX}^T)^{\owedge2}$, where the symbol $\owedge2$ means the elementwise square. In this case, $L=1$ and $\mbf{K}_1(\mbf{X})=p^{-1}\mbf{XX}^T$. Then note that
\begin{align*}
	f[\mbf{\Sigma}] & =\mbf{J}+\frac{2\xi_1}{p}\mbf{J}\odot\mbf{XX^T}+\frac{\xi_1^2}{p^2}(\mbf{XX}^T)^{\owedge2}\\
	& =\mbf{J}+2\xi_1\frac{1}{p}\mbf{XX}^T+\xi_1^2\frac{1}{p^2}(\mbf{XX}^T)^{\owedge2}.
\end{align*}
This shows that for polynomial output kernel and one input product kernel, $L'=3$ with $g_1(\xi_1)=1, g_2(\xi_1)=2\xi_1, g_3(\xi_1)=\xi_1^2$ and $h_1[\mbf{K}_1(\mbf{X})]=\mbf{J}, h_2[\mbf{K}_1(\mbf{X})]=\mbf{K}_1(\mbf{X})=p^{-1}\mbf{XX}^T, h_3[\mbf{K}_1(\mbf{X})]=p^{-2}(\mbf{XX}^T)^{\owedge2}$. The parameters $\theta_0,\ldots,\theta_L$ can be estimated via MINQUE as well. Based on the above discussion, we can see that the estimation of the variance components in KNN through MINQUE is an approximation. What we basically do here is to use a complex mixed model to approximate the KNN.

\section{Predictions}
\label{sec:pred}
In this section, we make a theoretical comparison of prediction performance between KNN and LMM. Based on our model, the best predictor for $\mbf{a}$ is given by
\begin{align*}
	\hat{\mbf{y}} & =\mbb{E}\left[\mbf{a}|\mbf{y}\right]=\mbb{E}\left[\mbb{E}\left(\mbf{a}|\mbf{y},\mbf{u}_1,\ldots,\mbf{u}_m\right)\right]\\
	& =\mbb{E}\left[\left(\sum_{j=1}^J\tau_j\mbf{K}_j(\mbf{U})\right)\left(\sum_{j=1}^J\tau_j\mbf{K}_j(\mbf{U})+\phi\mbf{I}_n\right)^{-1}\right]\mbf{y}\\
	& =\mbb{E}\left[\left(\sum_{j=1}^J\phi^{-1}\tau_j\mbf{K}_j(\mbf{U})\right)\left(\sum_{j=1}^J\phi^{-1}\tau_j\mbf{K}_j(\mbf{U})+\mbf{I}_n\right)^{-1}\right]\mbf{y}\\
	& :=\mbb{E}\left[\left(\sum_{j=1}^J\tilde{\tau}_j\mbf{K}_j(\mbf{U})\right)\left(\sum_{j=1}^J\tilde{\tau}_j\mbf{K}_j(\mbf{U})+\mbf{I}_n\right)^{-1}\right]\mbf{y},
\end{align*}
where $\tilde{\tau}_j=\tau_j\phi^{-1}, j=1,\ldots,m$. The prediction error based on $\hat{\mbf{y}}$ is given by
\begin{align*}
	R & =\left(\mbf{y}-\hat{\mbf{y}}\right)^T\left(\mbf{y}-\hat{\mbf{y}}\right)\\
	& =\mbf{y}^T\left(\mbf{I}_n-\mbb{E}\left[\left(\sum_{j=1}^J\tilde{\tau}_j\mbf{K}_j(\mbf{U})\right)\left(\sum_{j=1}^J\tilde{\tau}_j\mbf{K}_j(\mbf{U})+\mbf{I}_n\right)^{-1}\right]\right)^T\\
	& \hspace{3cm}\left(\mbf{I}_n-\mbb{E}\left[\left(\sum_{j=1}^J\tilde{\tau}_j\mbf{K}_j(\mbf{U})\right)\left(\sum_{j=1}^J\tilde{\tau}_j\mbf{K}_j(\mbf{U})+\mbf{I}_n\right)^{-1}\right]\right)\mbf{y}.
\end{align*}
Note that
\begin{align*}
	& \mbf{I}_n-\mbb{E}\left[\left(\sum_{j=1}^J\tilde{\tau}_j\mbf{K}_j(\mbf{U})\right)\left(\sum_{j=1}^J\tilde{\tau}_j\mbf{K}_j(\mbf{U})+\mbf{I}_n\right)^{-1}\right]\\
	= & \mbb{E}\left[\left(\sum_{j=1}^J\tilde{\tau}_j\mbf{K}_j(\mbf{U})+\mbf{I}_n-\sum_{j=1}^J\tilde{\tau}_j\mbf{K}_j(\mbf{U})\right)\left(\sum_{j=1}^J\tilde{\tau}_j\mbf{K}_j(\mbf{U})+\mbf{I}_n\right)^{-1}\right]\\
	= & \mbb{E}\left[\left(\sum_{j=1}^J\tilde{\tau}_j\mbf{K}_j(\mbf{U})+\mbf{I}_n\right)^{-1}\right],
\end{align*}
we have
$$
R=\mbf{y}^T\left(\mbb{E}\left[\left(\sum_{j=1}^J\tilde{\tau}_j\mbf{K}_j(\mbf{U})+\mbf{I}_n\right)^{-1}\right]\right)^2\mbf{y}.
$$
Direct evaluation of the prediction error $R$ is complicated. Instead, we approximate it based on asymptotic results. Same as the above, we focus on the case where $J=1$ and $\mbf{K}(\mbf{U})=f\left[\frac{1}{m}\mbf{UU}^T\right]$. The proof of the following Lemma can be found in Appendix \ref{Sec: proof}.

\begin{lemma}[Approximation of Prediction Error]\label{PredErrApprox}
	\begin{enumerate}[(i)]
		\item When $f(x)=x$, then as $m\to\infty$,
		$$
		R\simeq\mbf{y}^T\left(\sum_{l=1}^L\tilde{\tau}\xi\mbf{K}_l(\mbf{X})+\mbf{I}_n\right)^{-2}\mbf{y}.
		$$
		\item When $f$ is continuous and $f[\mbf{\Sigma}]\in\mcal{S}_+^n$, then as $m\to\infty$,
		$$
		R\simeq\mbf{y}^T\left(\tilde{\tau}f\left[\sum_{l=1}^L\xi\mbf{K}_l(\mbf{X})\right]+\mbf{I}_n\right)^{-2}\mbf{y}.
		$$
	\end{enumerate}
\end{lemma}

Now we compare the average prediction error between kernel neural network and linear mixed model. For a linear mixed model, the prediction error using best predictor can be obtained as follow. The proof can be found in Appendix \ref{Sec: proof}.
\begin{prop}[Prediction Error for a Linear Mixed Model]
	Consider the linear mixed effect model
	\begin{align*}
		\mbf{y} & =\mbf{a}+\mbf{\epsilon};\\
		\mbf{a} & \sim\mcal{N}_n\left(\mbf{0},\sigma_R^2\mbf{\Sigma}\right);\\
		\mbf{\epsilon} & \sim\mcal{N}_n\left(\mbf{0},\phi\mbf{I}_n\right).
	\end{align*}
	The prediction error based on quadratic loss and the best predictor $\hat{\mbf{y}}=\mbb{E}[\mbf{a}|\mbf{y}]=\tilde{\sigma}_R^2\mbf{\Sigma}\left(\tilde{\sigma}_R^2\mbf{\Sigma}+\mbf{I}_n\right)^{-1}\mbf{y}$ $(\tilde{\sigma}_R^2=\sigma_R^2\phi^{-1})$ is given by
	$$
	PE_{LMM}=\phi\sum_{i=1}^n\left(\tilde{\sigma}_R^2\lambda_i(\mbf{\Sigma})+1\right)^{-1},
	$$
	where $PE_{LMM}$ is the average prediction error for the linear mixed model.
\end{prop}

\begin{prop}\label{propKNN1}
	Assuming that $\sigma^2=\phi$ and $\tilde{\sigma}_R^2\leq\tilde{\tau}\min_{1\leq l\leq L}\xi_l$, we have
	$$
	PE_{KNN}\precsim PE_{LMM},
	$$ 
	where $PE_{KNN}$ stands for average prediction error for kernel neural network.
\end{prop}

\begin{proof}
	For the kernel neural network, the average prediction error is given by
	\begin{align*}
		PE_{KNN} & = \mbb{E}\left[\mbf{y}^T\mbf{A}^2\mbf{y}\right]=\mbb{E}\left[\mbb{E}\left(\mbf{y}^T\mbf{A}^2\mbf{y}|\mbf{u}_1,\ldots,\mbf{u}_m\right)\right]\\
		& =\phi\mbb{E}\left[\mrm{tr}\left(\mbf{A}^2\left(\tilde{\tau}\mbf{K}(\mbf{U})+\mbf{I}_n\right)\right)\right]\\
		& =\phi\mrm{tr}\left\{\left(\mbb{E}\left[\left(\tilde{\tau}\mbf{K}(\mbf{U})+\mbf{I}_n\right)^{-1}\right]\right)^2\mbb{E}\left[\tilde{\tau}\mbf{K}(\mbf{U})+\mbf{I}_n\right]\right\}\\
		& \simeq\phi\mrm{tr}\left\{\left(\tilde{\tau}\sum_{l=1}^L\xi_l\mbf{K}_l(\mbf{X})+\mbf{I}_n\right)^{-2}\left(\tilde{\tau}\sum_{l=1}^L\xi_l\mbf{K}_l(\mbf{X})+\mbf{I}_n\right)\right\}\\
		& =\phi\mrm{tr}\left\{\left(\tilde{\tau}\sum_{l=1}^L\xi_l\mbf{K}_l(\mbf{X})+\mbf{I}_n\right)^{-1}\right\}\\
		& \leq\phi\sum_{i=1}^n\left(\tilde{\tau}\min_{1\leq l\leq L}\xi_l\lambda_i\left(\sum_{l=1}^L\mbf{K}_l(\mbf{X})\right)+1\right)^{-1}
	\end{align*}
	Under the assumptions in this proposition and the linear mixed model with $\mbf{\Sigma}=\sum_{l=1}^L\mbf{K}_l(\mbf{X})$, we have
	$$
	\frac{\phi\left(\tilde{\tau}\min_{1\leq l\leq L}\lambda_i(\mbf{\Sigma})+1\right)^{-1}}{\sigma^2\left(\tilde{\sigma}_R^2\lambda_i(\mbf{\Sigma})+1\right)^{-1}}=\frac{\phi}{\sigma^2}\frac{\tilde{\sigma}_R^2\lambda_i(\mbf{\Sigma})+1}{\tilde{\tau}\min_{1\leq l\leq L}\xi_l\lambda_i(\mbf{\Sigma})+1}\leq1,
	$$
	which implies that $PE_{KNN}\precsim PE_{LMM}$.
\end{proof}

\begin{remark}
	The result for Proposition \ref{propKNN1} can be illustrated by using Figure \ref{propKNN1fig}. As shown in the Figure \ref{propKNN1fig} for the case of $L=1$, there are two paths from the kernel matrix based on $\mbf{X}$ to the response $\mbf{y}$. One is the kernel neural network path (solid line) and the other is the linear mixed model path (dash-dotted line). The intuition behind the assumption $\tilde{\sigma}_R^2\leq\tilde{\tau}\xi$ is that the kernel neural network should explain more variations than the linear mixed model as it has two portions.
	\begin{figure}[htbp]
		\centering
		\includegraphics[width=0.7\textwidth]{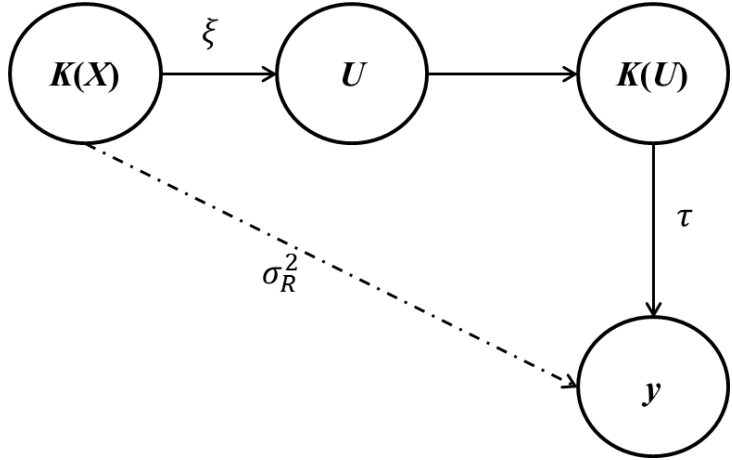}
		\caption{The intuition under the assumption $\tilde{\sigma}_R^2\leq\tilde{\tau}\xi$ in Proposition \ref{propKNN1}.}\label{propKNN1fig}
	\end{figure}
\end{remark}

We then extend the result to $\mbf{K}(\mbf{U})=f\left[\frac{1}{m}\mbf{UU}^T\right]$, where $f$ is as described in Lemma \ref{PredErrApprox}(ii).

\begin{prop}\label{propKNN2}
	Under the above notations, assuming that $\sigma^2=\phi$, $\tilde{\sigma}_R^2\leq\tilde{\tau}\min_{1\leq l\leq L}\xi_l$ and $\lambda_1\left((f-\iota)\left[\sum_{l=1}^L\xi_l\mbf{K}_l(\mbf{X})\right]\right)\geq0$ with $|f''(x)|\leq M$ for some $M>0$ and all $x$ between $\min_{i,j}\sigma_{ij}$ and $\max_{i,j}\frac{\mbf{v}_i^T\mbf{v}_j}{m}$, we have
	$$
	PE_{KNN}\precsim PE_{LMM},
	$$
	where $\lambda_1(\mbf{\Sigma})$ is the smallest eigenvalue of the matrix $\mbf{\Sigma}$.
\end{prop}

\begin{proof}
	Note that
	\begin{align*}
		PE_{KNN} & =\phi\mrm{tr}\left\{\left(\mbb{E}\left[\left(\tilde{\tau}\mbf{K}(\mbf{U})+\mbf{I}_n\right)^{-1}\right]\right)^2\mbb{E}\left[\tilde{\tau}\mbf{K}(\mbf{U})+\mbf{I}_n\right]\right\}\\
		& \simeq\phi\mrm{tr}\left\{\left(\tilde{\tau}f\left[\sum_{l=1}^L\xi_l\mbf{K}_l(\mbf{X})\right]+\mbf{I}_n\right)^{-1}\right\}\\
		& =\phi\sum_{i=1}^n\frac{1}{\tilde{\tau}\lambda_i\left(f\left[\sum_{l=1}^L\xi_l\mbf{K}_l(\mbf{X})\right]\right)+1}\\
		& \leq\phi\sum_{i=1}^n\frac{1}{\tilde{\tau}\lambda_i\left((f-\iota)\left[\sum_{l=1}^L\xi_l\mbf{K}_l(\mbf{X})\right]+\min_{1\leq l\leq L}\xi_l\sum_{l=1}^L\mbf{K}_l(\mbf{X})\right)+1},
	\end{align*}
	where $\iota:\mbf{\Sigma}\mapsto\mbf{\Sigma}$ is the identity map. Corollary 4.3.15 in \citet{horn_johnson_2012} implies that
	$$
	PE_{KNN}\precsim\phi\sum_{i=1}^n\frac{1}{\tilde{\tau}\min_{1\leq l\leq L}\xi_l\lambda_i(\mbf{K}_l(X))+\tilde{\tau}\lambda_1\left((f-\iota)\left[\sum_{l=1}^L\xi_l\mbf{K}_l(\mbf{X})\right]\right)+1}
	$$
\end{proof}

\begin{corollary}
	If $f\left[\sum_{l=1}^L\xi_l\mbf{K}_l(\mbf{X})\right]-\sum_{l=1}^L\xi_l\mbf{K}_l(\mbf{X})$ is positive semidefinite, then
	$$
	APEKNN\precsim APELMM.
	$$
\end{corollary}

\begin{exmp}[Polynomial Kernels]
	For a polynomial kernel of degree $d$, i.e., $\mbf{K}_{ij}(\mbf{U})=\left(c+\frac{\mbf{w}_i^T\mbf{w}_j}{m}\right)^d$, we have $f(x)=(c+x)^d=\sum_{k=0}^d\binom{d}{k}c^{d-k}x^k$ so that
	$$
	(f-\iota)(x)=c^d+(dc^{d-1}-1)x+\sum_{k=2}^d\binom{d}{k}c^{d-k}x^k.
	$$
	Theorem 4.1 in \citet{hiai2009monotonicity} states that for a real function on $(-\alpha,\alpha)$, $0<\alpha\leq\infty$, it is Schur positive\footnote{For a real function $f$ on $(-\alpha,\alpha)$ and for $n\in\mbb{N}$, it is Schur-positive of order $n$ if $f[\mbf{A}]$ is positive semidefinite for all positive semidefinite $\mbf{A}\in\mcal{M}_n(\mbb{R})$ with entries in $(-\alpha,\alpha)$.} if and only if it is analytic and $f^{(k)}(0)\geq0$ for all $k\geq0$. Since $f-\iota$ is a polynomial function, it is clearly analytic. We can then expand $f(x)$ using Taylor expansion around 0 and obtain
	$$
	\binom{d}{k}c^{d-k}=\frac{f^{(k)}(0)}{k!}\Rightarrow f^{(k)}(0)=\frac{d!}{(d-k)!}c^{d-k},\quad k=0,\ldots,d.
	$$
	Hence, we have
	$$
	(f-\iota)^{(k)}(x)=\left\{\begin{array}{ll}
		dc^{d-1}-1 & \mrm{if }k=1\\
		\frac{d!}{(d-k)!}c^{d-k} & \mrm{if }k\in\{0,1,\ldots,d\}\setminus\{1\}\\
		0 & k\geq d+1
	\end{array}\right..
	$$
	To make $f-\iota$ Schur positive, we only need to require $c\geq\sqrt[d-1]{\frac{1}{d}}\geq0$ so that the minimum eigenvalue condition of Proposition \ref{propKNN2} holds.
\end{exmp}

\section{Including Fixed Effects}
In the previous discussions, we focus on the case where the marginal distribution of the response variable $\mbf{y}$ has mean $\mbf{0}$. In many applications, as the one we see in the ADNI real data application, there are many important covariates which may have large effects to the response variable. In this part, we are going to extend the proposed KNN model to take the covariates into account. As we have mentioned in section \ref{sec:meth}, the general structure of KNN is
\begin{align*}
	\mbf{y}|\mbf{Z},\mbf{a} & \sim\mcal{N}_n\left(\mbf{Z\beta}+\mbf{a},\phi\mbf{I}_n\right)\\
	\mbf{a}|\mbf{u}_1,\ldots,\mbf{u}_m & \sim\mcal{N}_n\left(\mbf{0},\sum_{j=1}^J\tau_j\mbf{K}_j(\mbf{U})\right)\numberthis\label{eq:KNNwithCovariates}\\
	\mbf{u}_1,\ldots,\mbf{u}_m & \sim\mrm{ i.i.d. }\mcal{N}_n\left(\mbf{0},\sum_{l=1}^L\xi_l\mbf{K}_l(\mbf{X})\right).
\end{align*}
Similar to the situation considered before, given the latent variables $\mbf{u}_1,\ldots,\mbf{u}_m$, we have
$$
\mbf{y}|\mbf{Z},\mbf{u}_1,\ldots,\mbf{u}_m\sim\mcal{N}_n\left(\mbf{Z\beta},\sum_{j=1}^J\tau_j\mbf{K}_j(\mbf{U})+\phi\mbf{I}_n\right)
$$
and then the marginal mean and covariance matrix of $\mbf{y}$ can be obtained as follow:
\begin{align*}
	\mbb{E}\left[\mbf{y}|\mbf{Z}\right] & =\mbb{E}\left[\mbb{E}\left[\mbf{y}|\mbf{Z},\mbf{u}_1,\ldots,\mbf{u}_m\right]\right]=\mbb{E}\left[\mbf{Z\beta}\right]=\mbf{Z\beta}\\
	\mrm{Var}\left[\mbf{y}|\mbf{Z}\right] & =\mrm{Var}\left[\mbb{E}\left[\mbf{y}|\mbf{Z},\mbf{u}_1,\ldots,\mbf{u}_m\right]\right]+\mbb{E}\left[\mrm{Var}\left[\mbb{E}\left[\mbf{y}|\mbf{Z},\mbf{u}_1,\ldots,\mbf{u}_m\right]\right]\right]\\
	& =\mrm{Var}[\mbf{Z\beta}]+\mbb{E}\left[\sum_{j=1}^J\tau_j\mbf{K}_j(\mbf{U})+\phi\mbf{I}_n\right]\\
	& =\sum_{j=0}^J\tau_j\mbb{E}\left[\mbf{K}_j(\mbf{U})\right],
\end{align*}
where $\tau_0=\phi$ and $\mbf{K}_0(\mbf{U})=\mbf{I}_n$. Again, we focus on the case where $J=1$ and $\mbf{K}(\mbf{U})$ is of the form $f\left[\frac{1}{m}\mbf{UU}^T\right]$ with $f$ satisfying the generalized linear separable condition so that after suitable reparameterization, the variance components become estimable. A natural way to obtain the estimators for $\mbf{\beta}$ and the variance components $\mbf{\theta}$ is to first obtain a ``good" estimates for the variance components $\hat{\mbf{\theta}}$ and then plug it into the Aitken equation to obtain the estimates for fixed-effect parameters $\hat{\mbf{\beta}}$. \citet{kackar1981unbiasedness} showed that as long as $\hat{\mbf{\theta}}$ is even and translation-invariant, $\mbf{p}^T\hat{\mbf{\beta}}+\mbf{q}^T\hat{\mbf{a}}$ is an unbiased prediction for the quantity $\mbf{p}^T\mbf{\beta}+\mbf{q}^T\mbf{a}$.

Let $\mbf{R}$ be an $r\times n$ matrix with $r=n-\mrm{rank}(\mbf{Z})$ such that $\mbf{RZ}=\mbf{O}$ and $\mbf{RR}^T=\mbf{I}_r$. Such a matrix can be found using the QR decomposition of $\mbf{Z}$ \citep{glmms}. The estimators for the variance components can then be estimated based on the transformed model:
\begin{align*}
	\tilde{\mbf{y}}|\tilde{\mbf{a}} & \sim\mcal{N}_r(\tilde{\mbf{a}},\phi\mbf{I}_r)\\
	\tilde{\mbf{a}}|\mbf{u}_1,\ldots,\mbf{u}_m & \sim\mcal{N}_r\left(\mbf{0},\sum_{j=1}^J\tau_j\mbf{R}\mbf{K}_j(\mbf{U})\mbf{R}^T\right)\\
	\mbf{u}_1,\ldots,\mbf{u}_m & \sim\mrm{ i.i.d. }\mcal{N}_n\left(\mbf{0},\sum_{l=1}^L\xi_l\mbf{K}_l(\mbf{X})\right),
\end{align*}
where $\tilde{\mbf{y}}=\mbf{Ry}$ and $\tilde{\mbf{a}}=\mbf{Ra}$, which is the same model framework we mainly discussed in section \ref{sec:Methodologies}. As we have seen, the estimators for the variance components after reparameterization is of the form $\tilde{\mbf{y}}^T\mbf{\Theta}\tilde{\mbf{y}}=\mbf{y}^T\mbf{R^T\Theta Ry}$. Since it is a quadratic form, clearly it is an even function in $\mbf{y}$. To see that it is translation invariant, we note that for any $\mbf{c}\in\mcal{C}(\mbf{Z})$, we can know that $\mbf{c}=\mbf{Zb}$ for some vector $\mbf{b}\in\mbb{R}^r$ and we have
\begin{align*}
	(\mbf{y}-\mbf{\gamma})^T\mbf{R}^T\mbf{\Theta R}(\mbf{y}-\mbf{\gamma}) & =(\mbf{y}-\mbf{Zb})^T\mbf{R}^T\mbf{\Theta R}(\mbf{y}-\mbf{Zb})\\
	& =\mbf{y}^T\mbf{R}^T\mbf{\Theta Ry}-2\mbf{y}^T\mbf{R}^T\mbf{\Theta RZb}+\mbf{b}^T\mbf{Z}^T\mbf{R}^T\mbf{\Theta RZb}\\
	& =\mbf{y}^T\mbf{R}^T\mbf{\Theta Ry},
\end{align*}
where the last equality follows since $\mbf{RZ}=\mbf{O}$ as defined. Therefore, we can know that the obtained estimators for variance components are also translation-invariant and based on the results in \citet{kackar1981unbiasedness}, the obtained estimator for $\mbf{\beta}$ by plugging in the variance component estimators is unbiased.

For the prediction error, we note that when the covariates present, the predictor for $\mbf{y}$ is given by $\hat{\mbf{y}}=\mbf{Z}\hat{\mbf{\beta}}+\hat{\mbf{a}}=\mbf{Z}\hat{\mbf{\beta}}+\mbb{E}[\mbf{a}|\mbf{y}]$. Based on the result from linear mixed models, we know that
\begin{align*}
	\hat{\mbf{a}} & =\mbb{E}\left[\left(\sum_{j=1}^J\tilde{\tau}_j\mbf{K}_j(\mbf{U})\right)\left(\sum_{j=1}^J\tilde{\tau}_j\mbf{K}_j(\mbf{U})+\mbf{I}_n\right)^{-1}\right]\left(\mbf{y}-\mbf{Z}\mbf{\beta}\right)\\
	& =\mbb{E}\left[\left(\sum_{j=1}^J\tilde{\tau}_j\mbf{K}_j(\mbf{U})\right)\left(\sum_{j=1}^J\tilde{\tau}_j\mbf{K}_j(\mbf{U})+\mbf{I}_n\right)^{-1}\right]\left[\mbf{I}_n-\mbf{Z}\left(\mbf{Z}^T\mbf{V}^{-1}\mbf{Z}\right)^-\mbf{Z}^T\mbf{V}^{-1}\right]\mbf{y}
\end{align*}
where $\mbf{V}=\mrm{Var}[\mbf{y}]=\sum_{j=1}^J\tau_j\mbb{E}[\mbf{K}_j(\mbf{U})]+\phi\mbf{I}_n$ and then
\begin{align*}
	\mbf{y}-\hat{\mbf{y}} & =\left[\mbf{I}_n-\mbf{Z}\left(\mbf{Z}^T\mbf{V}^{-1}\mbf{Z}\right)^-\mbf{Z}^T\mbf{V}^{-1}\right]\mbf{y}\\
	& \hspace{1cm}-\mbb{E}\left[\left(\sum_{j=1}^J\tilde{\tau}_j\mbf{K}_j(\mbf{U})\right)\left(\sum_{j=1}^J\tilde{\tau}_j\mbf{K}_j(\mbf{U})+\mbf{I}_n\right)^{-1}\right]\left[\mbf{I}_n-\mbf{Z}\left(\mbf{Z}^T\mbf{V}^{-1}\mbf{Z}\right)^-\mbf{Z}^T\mbf{V}^{-1}\right]\mbf{y}\\
	& =\left(\mbf{I}_n-\mbb{E}\left[\left(\sum_{j=1}^J\tilde{\tau}_j\mbf{K}_j(\mbf{U})\right)\left(\sum_{j=1}^J\tilde{\tau}_j\mbf{K}_j(\mbf{U})+\mbf{I}_n\right)^{-1}\right]\right)\left[\mbf{I}_n-\mbf{Z}\left(\mbf{Z}^T\mbf{V}^{-1}\mbf{Z}\right)^-\mbf{Z}^T\mbf{V}^{-1}\right]\mbf{y}\\
	& =\mbb{E}\left[\left(\sum_{j=1}^J\tilde{\tau}_j\mbf{K}_j(\mbf{U})+\mbf{I}_n\right)^{-1}\right]\left[\mbf{I}_n-\mbf{Z}\left(\mbf{Z}^T\mbf{V}^{-1}\mbf{Z}\right)^-\mbf{Z}^T\mbf{V}^{-1}\right]\mbf{y},
\end{align*}
where the last equality follows by noting that 
$$
\mbf{I}_n-\mbb{E}\left[\left(\sum_{j=1}^J\tilde{\tau}_j\mbf{K}_j(\mbf{U})\right)\left(\sum_{j=1}^J\tilde{\tau}_j\mbf{K}_j(\mbf{U})+\mbf{I}_n\right)^{-1}\right]=\mbb{E}\left[\left(\sum_{j=1}^J\tilde{\tau}_j\mbf{K}_j(\mbf{U})+\mbf{I}_n\right)^{-1}\right]
$$ 
as shown above. Therefore, by letting $\mbf{A}=\mbb{E}\left[\left(\sum_{j=1}^J\tilde{\tau}_j\mbf{K}_j(\mbf{U})+\mbf{I}_n\right)^{-1}\right]$ and $\mbf{P}_{\mbf{V}}=\mbf{Z}\left(\mbf{Z}^T\mbf{V}^{-1}\mbf{Z}\right)^-\mbf{Z}^T\mbf{V}^{-1}$, the prediction error is obtained as follow:
\begin{align*}
	(\mbf{y}-\hat{\mbf{y}})^T(\mbf{y}-\hat{\mbf{y}}) & =\mbf{y}^T\left(\mbf{I}_n-\mbf{P}_{\mbf{V}}^T\right)\hat{\mbf{A}}^2\left(\mbf{I}_n-\mbf{P}_{\mbf{V}}\right)\mbf{y}.
\end{align*}
When $J=1$, as we have shown in the proof of Lemma \ref{PredErrApprox}, $\mbb{E}\left[\left(\tilde{\tau}\mbf{K}(\mbf{U})+\mbf{I}_n\right)^{-1}\right]\simeq$\\$\left(\tilde{\tau}f\left[\sum_{l=1}^L\xi_l\mbf{K}_l(\mbf{X})\right]+\mbf{I}_n\right)^{-1}$ so that we can estimate the prediction error by using simple plug-in estimators.
\begin{align*}
	\widehat{(\mbf{y}-\hat{\mbf{y}})^T(\mbf{y}-\hat{\mbf{y}})} & \simeq\mbf{y}^T\left(\mbf{I}_n-\mbf{P}_{\hat{\mbf{V}}}^T\right)\left(\hat{\tilde{\tau}}f\left[\sum_{l=1}^L\hat{\xi}_l\mbf{K}_l(\mbf{X})\right]+\mbf{I}_n\right)^{-2}\left(\mbf{I}_n-\mbf{P}_{\hat{\mbf{V}}}\right)\mbf{y}.
\end{align*}

\section{Simulations}
\label{sec:sim}
In this section, we conducted some simulations to compare the prediction performance of KNN with MINQUE estimation to the prediction performance of Best Linear Unbiased Estimator (BLUP) in linear mixed model. All the simulations are based on 100 individuals with 500 Monte Carlo iterations.

\subsection{Nonlinear Random Effect}
In this simulation, we examine the performances of both methods under the situation of nonlinear random effects. Specifically, we used the following model to simulate the response:
\begin{equation}\label{simNonlinear}
	\mbf{y}=\mbf{1}_n+2\mbf{\zeta}+f(\mbf{a})+\mbf{\epsilon},\quad\mbf{a}\sim\mcal{N}_n\left(\mbf{0},\frac{1}{p}\mbf{GG}^T\right),
\end{equation}
where $\mbf{G}$ is an $n\times p$ matrix containing the genetic information (SNP) and $\mbf{\zeta},\mbf{\epsilon}\sim\mrm{ i.i.d. }\mcal{N}_n(\mbf{0},\mbf{I}_n)$. In the simulation, four types of functions $f$ are considered, which are linear ($f(x)=x$), sine ($f(x)=\sin(2\pi x)$), inverse logistic ($f(x)=1/(1+e^{-x})$) and polynomial function of order 2 ($f(x)=x^2$). When applying the kernel neural network, we set $L=J=1$ and choose $\mbf{K}(\mbf{X})$ and $\mbf{K}(\mbf{U})$ as either product kernel or polynomial kernel of order 2. We summarize the prediction errors of LMM and KNN via boxplot in Figure \ref{nonlinear1}. Figure \ref{nonlinear1} shows the results when $f$ is a linear or a sine function. The cases where $f$ is an inverse logistic function or a polynomial function of order 2 are summarized in Appendix C. As we can observe from the boxplots, when the output kernel is chosen to be the product kernel, the performance of KNN is similar to the performance of LMM although when the input kernel is chosen to be polynomial kernel, KNN gets a slightly better prediction error, which we think is not significantly different from that of LMM. However, KNN performs significantly better than LMM when the output kernel is chosen to be polynomial kernel. As one can tell from the box plots, when both the input and output kernels are chosen to be polynomial, the KNN has the best performance in terms of the prediction error, which is consistent for all nonlinear functions simulated in this section.

\begin{figure}[htbp]
	\centering
	\includegraphics[width=0.45\textwidth, height=0.45\textheight]{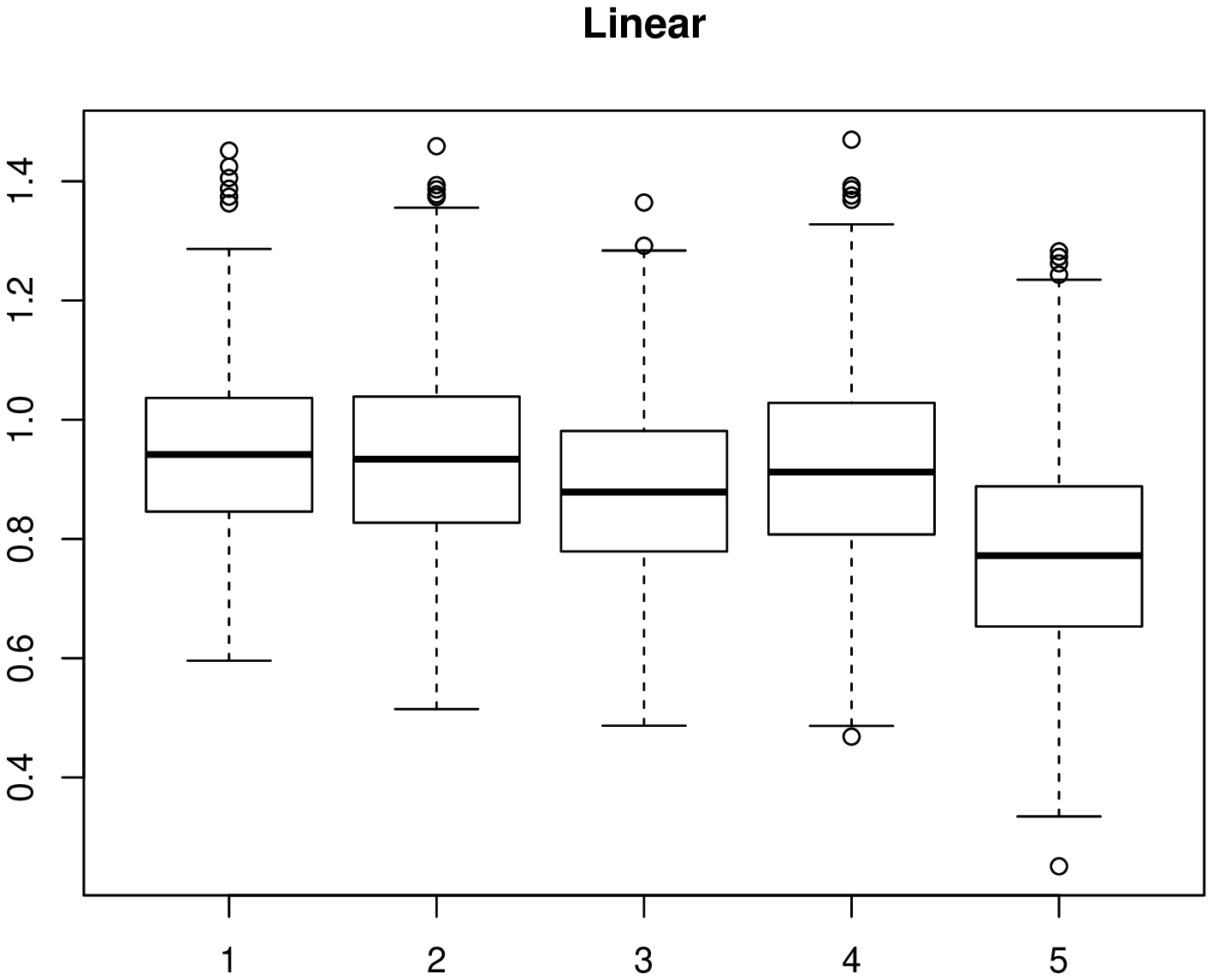}
	\includegraphics[width=0.45\textwidth, height=0.45\textheight]{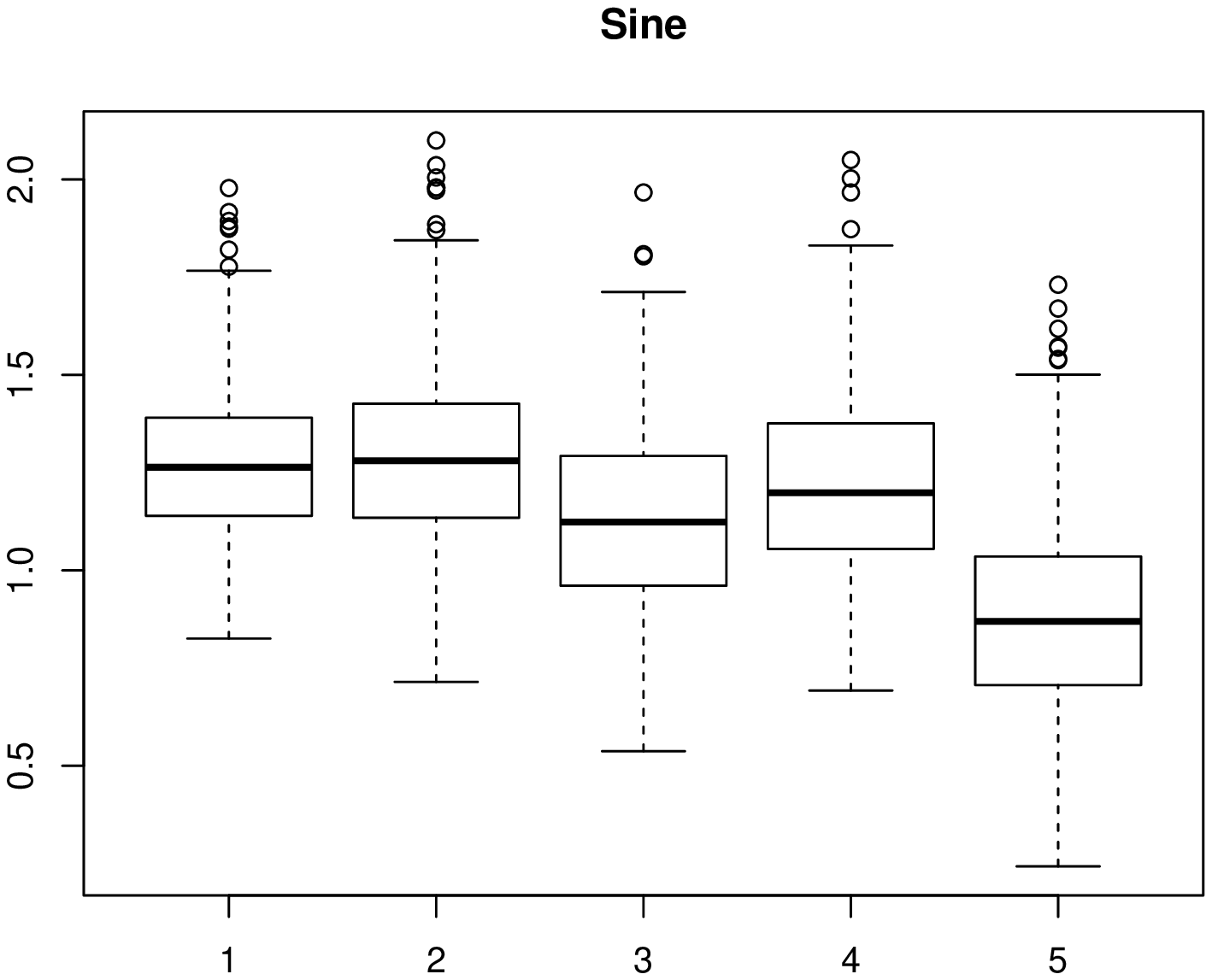}
	\caption{The boxplots for linear mixed models (LMM) and kernel neural networks (KNN) in terms of prediction errors. The left panel shows the results when a linear function is used and the right panel shows the results when a sine function is used. In the horizontal axis, ``1" corresponds to the LMM; ``2" corresponds to the KNN with product input kernel and product output kernel; ``3" corresponds to the KNN with product input and polynomial output; ``4" corresponds to the KNN with polynomial input and product output and ``5" corresponds to the polynomial input and polynomial output.}\label{nonlinear1}
\end{figure}

\subsection{Nonadditive Effects}
In this simulation, we explore the performances of both method under non-additive effects. We conducted two simulations in terms of two different types of non-additive effects. In the first simulation, we mainly focus on the interaction effect and generate the response using the following model:
$$
\mbf{y}=f(\mbf{G})+\mbf{\epsilon},
$$
where $\mbf{G}=[\mbf{g}_1,\ldots,\mbf{g}_p]\in\mbb{R}^{n\times p}$ is the SNP data and $\mbf{\epsilon}\sim\mcal{N}_n(\mbf{0},\mbf{I}_n)$. When applying both methods, the mean is adjusted so that the response has marginal mean 0. In the simulation, we randomly pick 10 causal SNPs, denoted by $\mbf{g}_{i_1},\ldots,\mbf{g}_{i_{10}}$ and consider  the following function,
$$
f(\mbf{G})=\sum_{1\leq j_1<j_2\leq 10}\mbf{g}_{i_{j_1}}\odot\mbf{g}_{i_{j_2}},
$$
where $\odot$ stands for the Hadamard product. For LMM, the product kernel was used as the covariance matrix to generate the random effect. The result is shown in Figure \ref{NAIN}. It is interesting to notice from the boxplots that both LMM and KNN have many outliers when product kernel was used. Overall, LMM has larger variations compared to KNN in this simulation. When the output kernel in KNN is the polynomial kernel, the performance of KNN is much better than that of LMM.

\begin{figure}[htbp]
	\centering
	\includegraphics[height=0.6\textheight, width=\textwidth]{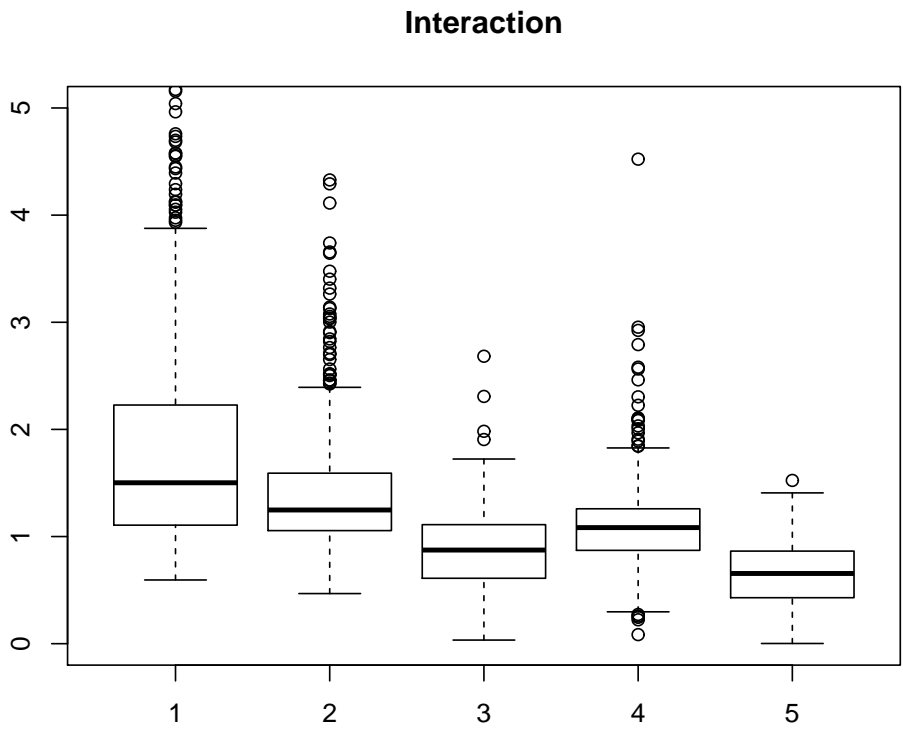}
	\caption{The boxplots for linear mixed models (LMM) and kernel neural networks (KNN) in terms of prediction errors based on the simulation model focusing on the interaction effect. The vertical axis is scaled to 0-5 by removing some outliers to make the comparison visually clear. In the horizontal axis, ``1" corresponds to the LMM; ``2" corresponds to the KNN with product input kernel and product output kernel; ``3" corresponds to the KNN with product input and polynomial output; ``4" corresponds to the KNN with polynomial input and product output and ``5" corresponds to the polynomial input and polynomial output.}\label{NAIN}
\end{figure}

There are three main modes of inheritance: additive, dominant and recessive. In many situations, the additive coding (AA=0, Aa=1, aa=2) is used. In the second simulation, we consider the dominant coding (AA=1, Aa=1, aa=0) and the recessive coding (AA=0, Aa=1, aa=1). The response was simulated based on the model:
$$
\mbf{y}=\mbf{a}+\mbf{\epsilon},\quad\mbf{a}\sim\mcal{N}_n\left(\mbf{0},\frac{1}{p}\mbf{G}'\mbf{G}'^T\right),\quad\mbf{\epsilon}\sim\mcal{N}_n(\mbf{0},\mbf{I}_n),
$$
where $\mbf{G}'$ is a SNP data matrix based on dominant coding or recessive coding so that each element in $\mbf{G}'$ takes only two possible values 0 and 1. Figure \ref{NADCRC} summarizes the simulation results. By comparing the two boxplots in Figure \ref{NADCRC}, the performances look similar in both cases. Similar as before, KNN with input polynomial kernel and output polynomial kernel achieves the lowest prediction error.

\begin{figure}[htbp]
	\centering
	\includegraphics[width=0.45\textwidth]{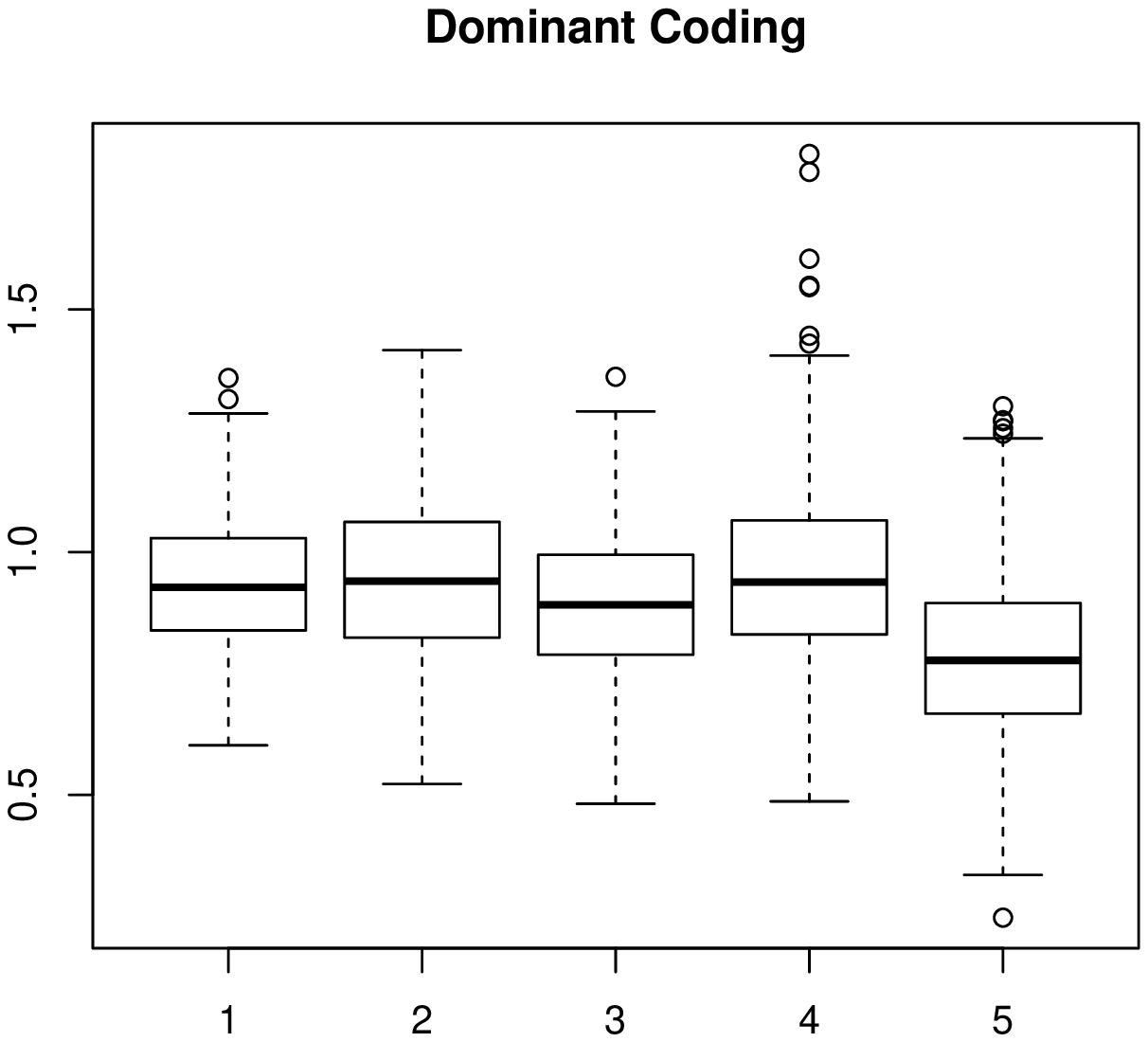}
	\includegraphics[width=0.45\textwidth]{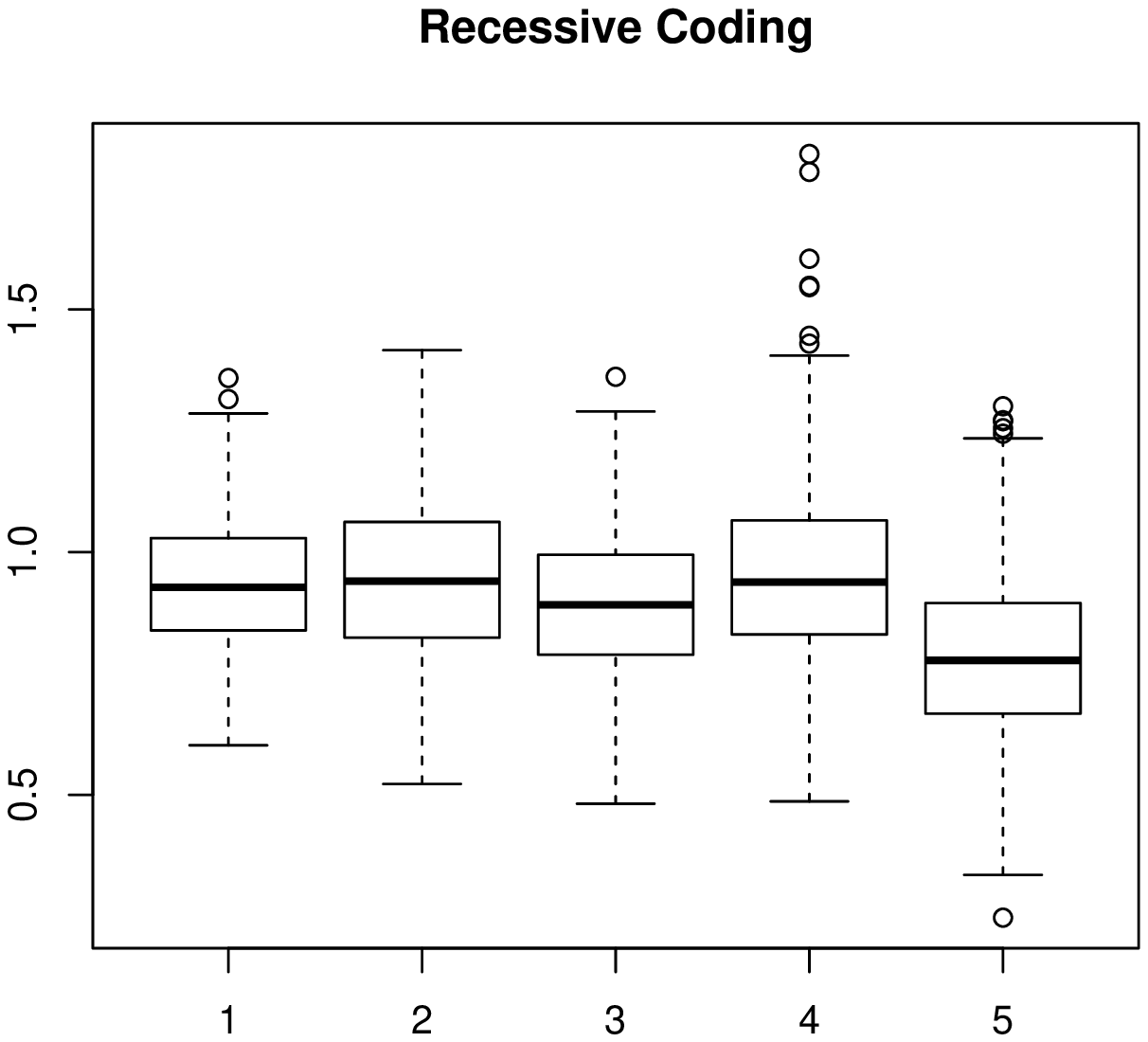}
	\caption{The boxplots for linear mixed models (LMM) and kernel neural networks (KNN) in terms of prediction errors based on the simulation model using dominant coding (left figure) and recessive coding (right figure) for SNPs. In the horizontal axis, ``1" corresponds to the LMM; ``2" corresponds to the KNN with product input kernel and product output kernel; ``3" corresponds to the KNN with product input and polynomial output; ``4" corresponds to the KNN with polynomial input and product output and ``5" corresponds to the polynomial input and polynomial output.}\label{NADCRC}
\end{figure}

\subsection{Non-normal Error Distributions}
In this simulation, we consider different types of error distributions and explore the performance in terms of prediction error between LMM and KNN. Specifically, we focus on two types of error distributions. The first one is a $t$-distribution with degrees of freedom 2, which is a heavy-tailed distribution and the second one is a centered $\chi_1^2$-distribution, which is a non-symmetric distribution. The response was simulated based on the model as in section 5.1 except that the function applied to the random effect is $f(x)=x$ and the distribution for the error term $\epsilon$ is either $t_2$ or centered $\chi_1^2$ distribution. The results are summarized in Figure \ref{nonNormal}. From the results, we can know that the KNN with polynomial input and output kernel is slightly more robust to non-normal distributions compared with LMM and KNN with other combinations of input and output kernels considered in the simulations. Since $t_2$-distribution is a heavy-tailed distribution, we can find that there are more outliers than in the normal case. In terms of the number of outliers, KNN with polynomial input and output kernel still performs the best.

\begin{figure}[htbp]
	\centering
	\includegraphics[width=0.45\textwidth]{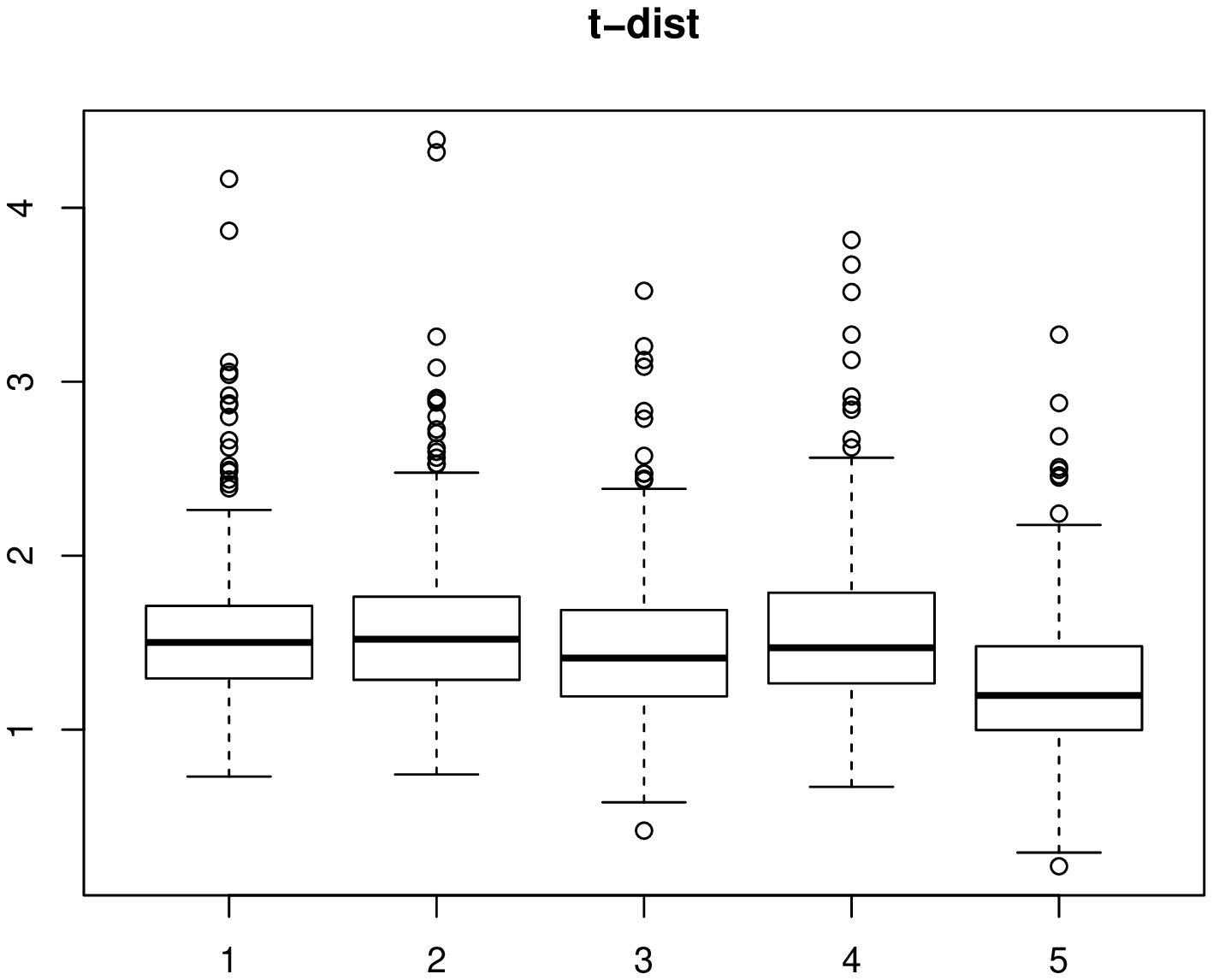}
	\includegraphics[width=0.45\textwidth]{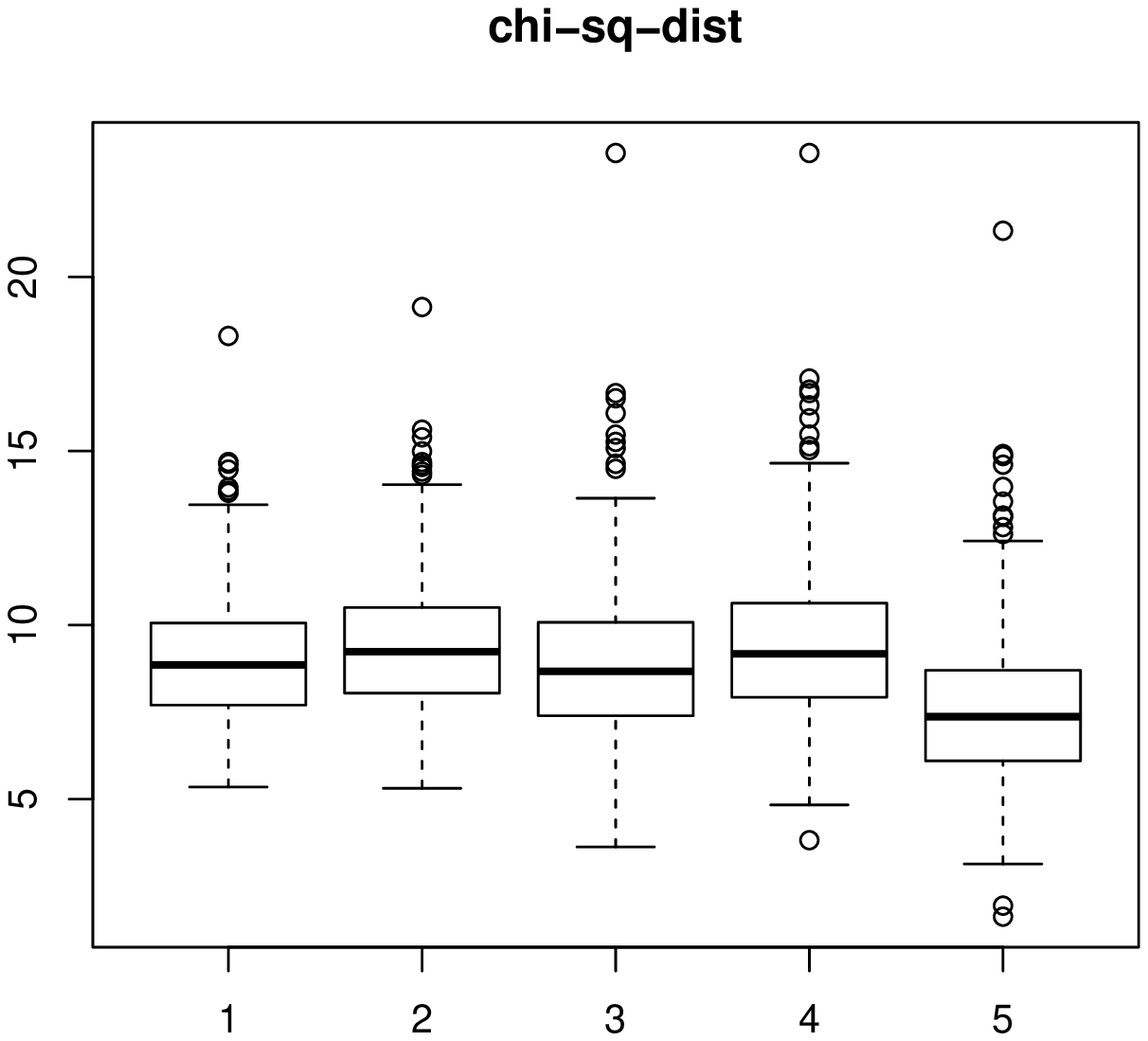}
	\caption{The boxplots for linear mixed models (LMM) and kernel neural networks (KNN) in terms of prediction errors based on the simulation model using $t$-distribution for error (left figure) and centered $\chi_1^2$-distribution for error (right figure). In the horizontal axis, ``1" corresponds to the LMM; ``2" corresponds to the KNN with product input kernel and product output kernel; ``3" corresponds to the KNN with product input and polynomial output; ``4" corresponds to the KNN with polynomial input and product output and ``5" corresponds to the polynomial input and polynomial output.}\label{nonNormal}
\end{figure}

\section{Real Data Application}
\label{sec:realData}
We applied our method to the whole genome sequencing data from Alzheimer's Disease Neuroimaging Initiative (ADNI) as well and made predictions on the responses. A total of 808 individuals at the baseline of the ADNI1 and ADNI2 studies have the whole genome sequencing data. We dropped the single nucleotide polymorphisms (SNPs) with low calling rate ($<$0.9), or low minor allele frequencies (MAF) ($<$0.01), or those failed to pass the Hardy Weinberg exact test ($p$-value$<$1e-6), and non-European American samples were also dropped. The data was then uploaded to the server in the University of Michigan for posterior likelihood imputation (\url{https://imputationserver.sph.umich.edu/index.html}). From the imputed data, we extracted SNPs with allelic $R^2>0.9$ and then the covariance kernel matrix the normalized identity-by-state (IBS) kernel matrix were constructed for analysis. Specifically, the $(i,j)$the element in each of the two kernel matrix is defined as follows:
\begin{align*}
	\mbf{K}^{\mrm{cov}}(\mbf{g}_i,\mbf{g}_j) & =\frac{1}{p-1}\sum_{k=1}^p\left(g_{ik}-\frac{1}{p}\sum_kg_{ik}\right)\left(g_{jk}-\frac{1}{p}\sum_kg_{jk}\right)\\
	\mbf{K}^{\mrm{ibs}}(\mbf{g}_i,\mbf{g}_j) & =\frac{1}{2p}\sum_{k=1}^p\left[2-|g_{ik}-g_{jk}|\right],
\end{align*}
where $p$ is the number of SNPs in all expressions.

Four volume measures of cortical regions, which are hippocampus, ventricles, entorhinal and whole brain volumes were used as phenotypes of interest. We chose these four cortical regions since they play important roles in the Alzheimer's disease (AD). The loss in the volumes of the whole brain, hippocampus and entorhinal and the increment in the ventricular volume can be detected among AD patients. When we applied both the KNN method and LMM method, we only include the subjects having both genetic information and phenotypic information, which results in 513 individuals for hippocampus; 564 individuals for ventricles; 516 individuals for entorhinal and 570 individuals for whole brain volumes. 

The response variable was chosen to be the natural logarithm of the volumes of the four cortical regions and the covariates were chosen as the age, gender, education status and \textit{APOE4}. Then the KNN is based on the model
\begin{align*}
	\mbf{y}|\mbf{u}_1,\ldots,\mbf{u}_m & \sim\mcal{N}_n\left(\mbf{Z\beta},\tau\mbf{K}(\mbf{U})+\phi\mbf{I}_n\right)\\
	\mbf{u}_1,\ldots,\mbf{u}_m & \sim\mcal{N}_n\left(\mbf{0},\xi_1\mbf{K}^{\mrm{cov}}+\xi_2\mbf{K}^{\mrm{ibs}}\right),
\end{align*}
and the LMM is based on the following model assumption:
\begin{equation}\label{adniModel}
	\mbf{y}\sim\mcal{N}_n\left(\mbf{Z\beta},\tau_1\mbf{K}^{\mrm{cov}}+\tau_2\mbf{K}^{\mrm{ibs}}+\tau_3\mbf{I}_n\right).
\end{equation}
The restricted maximum likelihood estimates for $\tau_i$, $i=1,2,3$ were then calculated based on the Fisher scoring methods and the BLUP for the LMM were computed based on these estimators. Similarly, when we applied the KNN methods, the two kernel matrices $\mbf{K}^{\mrm{cov}}$ and $\mbf{K}^{\mrm{ibs}}$ were used as the input kernel matrices and the output kernel matrix is chosen to be either the product kernel or the polynomial kernel of order 2. The average mean square errors on predictors of both methods were summarized in the Table \ref{tab:mseADNI}.

\begin{table}[htbp]
	\caption{Average mean squared prediction error of KNN with product output kernel matrix, KNN with output kernel matrix as polynomial of order 2 and the BLUP based on LMM.  \label{tab:mseADNI}}
	\begin{center}
		\begin{tabular}{c|ccc}
			\hline
			& KNN(prod) & KNN(poly) & LMM\\
			\hline
			Hippocampus & 2.01e-06 & 4.17e-05 & 2.11e-02 \\
			Ventricles & 2.44e-03 & 1.98e-03 & 2.24e-01 \\
			Entorhinal & 1.01e-05 & 1.06e-04 & 4.05e-02 \\
			Whole Brain Volume & 1.25e-07 & 3.04e-08 & 3.68e-11 \\
			\hline
		\end{tabular}
	\end{center}
\end{table}

As we can see from the table, both KNN and LMM have good prediction errors. However, we would say that the prediction errors from KNN are more realistic. It is interesting to note that the average prediction errors of LMM for entorhinal and whole brain volume are extremely close to zero. We checked the estimates of the variance components in this case and noticed that the variance component associated with the identity matrix, which is $\tau_3$ as in (\ref{adniModel}). Since the BLUP of the random effect in this case is given by
\begin{align*}
	BLUP & =\mbf{Z}\hat{\mbf{\beta}}+\left(\tau_1\mbf{K}^{\mrm{cov}}+\tau_2\mbf{K}^{\mrm{ibs}}\right)\left(\tau_1\mbf{K}^{\mrm{cov}}+\tau_2\mbf{K}^{\mrm{ibs}}+\tau_3\mbf{I}_n\right)^{-1}(\mbf{y}-\mbf{Z}\hat{\mbf{\beta}})
\end{align*}
so that if $\tau_3$ is very close to 0, we can know that the BLUP is very close to the response $\mbf{y}$, which leads to the extremely low prediction error. On the other hand, for KNN, since when the estimate of error variance becomes negative, we project the MINQUE matrix onto the positive semidefinite cone so that we will always get a positive estimate for the error variance component, which makes the calculation of prediction error more reasonable.

\section{Discussion}
\label{sec:discussion}
In this paper, a kernel-based neural network model was proposed for prediction analyses. The kernel-based neural network can be thought of as an extension of linear mixed model since it can reduce to a linear mixed model through choosing product kernel matrix as the output kernel matrix and via reparameterization. A modified MINQUE strategy is used to obtain the estimators of variance components in the KNN model if the output kernel satisfies the generalized linear separable condition. Empirical simulation studies and real data application show that the KNN model can achieve better performance in terms of the mean squared prediction error when the output kernel matrix is chosen as the polynomial kernel matrix. This is analogous to the popular neural network model, where better prediction accuracy can be achieved when nonlinear activation functions are applied. 

Many extensions can be made to make the KNN model more flexible and have much broader applications. First, it may be possible to consider conducting base kernel matrix selection. Although in this paper, we do not consider how to choose the number of base kernel matrix $L$, but too many kernel matrices will certainly increase the amount of redundant information. So it may be beneficial to propose a criterion on choosing the base kernel matrices. An other possible extension of the KNN model is more challenging. The theoretical properties discussed in this paper mainly focus on the case where only one output kernel matrix is used and the kernel function has a specific form. It is natural to consider more general kernel functions, but the estimation procedure of the variance component would be more complex. Moreover, it is also advisable to consider the performance of KNN if deep network structures were applied.

\newpage
\begin{appendices}
	\section{Some Results from Concentration Inequality and Matrix Analysis}
	\subsection{Sub-Gaussian and Sub-Exponential Inequalities}
	In this part, we present two basic concentration inequalities used in the main text. For more details on concentration inequality, readers may refer to \citet{buldygin2000metric}, \citet{boucheron2004concentration} and  \citet{ledoux2005concentration}.
	
	\begin{mydef}
		\begin{enumerate}[(i)]
			\item A random variable $X$ with mean $\mu=\mbb{E}[X]$ is called sub-Gaussian if there exists a number $\sigma\geq0$ such that
			$$
			\mbb{E}\left[e^{\lambda(X-\mu)}\right]\leq\exp\left\{\frac{1}{2}\lambda^2\sigma^2\right\},\quad\forall\lambda\in\mbb{R}
			$$
			The constant $\sigma$ is known as the sub-Gaussian parameter.
			
			\item A random variable $X$ with mean $\mu=\mbb{E}[X]$ is called sub-exponential (also called pre-Gaussian) if there exist non-negative constants $(\beta,\alpha)$ such that
			$$
			\mbb{E}\left[e^{\lambda(X-\mu)}\right]\leq\exp\left\{\frac{1}{2}\lambda^2\beta^2\right\},\quad\forall|\lambda|<\frac{1}{\alpha}.
			$$
			The pair of constants $(\beta,\alpha)$ is known as the sub-exponential parameter.
		\end{enumerate}
	\end{mydef}
	
	\noindent
	Based on the definition of sub-Gaussian random variables, it is clear that if $X\sim\mcal{N}(\mu,\sigma^2)$, then $X$ is sub-Gaussian with sub-Gaussian parameter $\sigma$. The tail probabilities of both sub-Gaussian and sub-exponential random variables can be bounded exponentially. For sub-Gaussian random variables, the result is the famous Hoeffding inequality.
	
	\begin{theorem}\citep{hoeffding1963probability}\label{HoeffdingIneq}
		Suppose that the random variables $X_1,\ldots,X_n$ are independent and $X_i$ has mean $\mu_i$ and sub-Gaussian parameter $\sigma_i$. Then for all $t>0$,
		$$
		\mbb{P}\left(\left|\sum_{i=1}^n(X_i-\mu_i)\right|>t\right)\leq2\exp\left\{-\frac{t^2}{2\sum_{i=1}^n\sigma_i^2}\right\}.
		$$
	\end{theorem}
	
	\begin{theorem}\label{TailBoundForSubExp}
		Suppose that $X$ is a sub-exponential random variable with mean $\mu$ and sub-exponential parameters $(\beta,\alpha)$. Then for all $t>0$,
		$$
		\mbb{P}\left(|X-\mu|>t\right)\leq\left\{\begin{array}{ll}
			2\exp\left\{-\frac{t^2}{2\beta^2}\right\}, & \mrm{if }0<t\leq\frac{\beta^2}{\alpha}\\
			2\exp\left\{-\frac{t}{2\alpha}\right\}, & \mrm{if }t>\frac{\beta^2}{\alpha}
		\end{array}
		\right..
		$$
	\end{theorem}
	
	\newpage
	\subsection{Results on Matrix Analysis}
	Proposition \ref{GentleProp} shows that a inverse map of a matrix is a continuous map, which will be frequently used in later parts when we approximate the average prediction errors.
	\begin{prop}[\citet{gentle2007matrix}]\label{GentleProp}
		Let $\mbf{\Psi}$ be an arbitrary invertible matrix. Then the map $f:\mbf{\Psi}\mapsto\mbf{\Psi}^{-1}$ is continuous.
	\end{prop}
	
	\begin{proof}
		Since $\mbf{\Psi}^{-1}=|\mbf{\Psi}|^{-1}\mrm{Adj}(\mbf{\Psi})$, where $\mrm{Adj}(\mbf{\Psi})$ is the adjugate matrix of $\mbf{\Psi}$, i.e., $\mrm{Adj}(\mbf{\Psi})=\mbf{C}^T$ and $\mbf{C}$ is the cofactor matrix of $\mbf{\Psi}$ with $\mbf{C}_{ij}=(-1)^{i+j}\mbf{M}_{ij}$ and $\mbf{M}_{ij}$ is the $(i,j)$ cofactor of $\mbf{\Psi}$. Based on the definition of determinant, it is easy to see that the map $g:\mbf{\Psi}\mapsto|\mbf{\Psi}|$ is continuous and the map $h:\mbf{\Psi}\mapsto\mrm{Adj}(\mbf{\Psi})$ is continuous as well. Therefore, the map $f:\mbf{\Psi}\mapsto\mbf{\Psi}^{-1}$ is continuous.
	\end{proof}
	
	Another result that will be used later is that any two matrix norms are equivalent in the sense that for any given pair of matrix norms $\|\cdot\|_s$ and $\|\cdot\|_t$, there is a finite positive constant $C_{st}$ such that
	$$
	\|\mbf{A}\|_s\leq C_{st}\|\mbf{A}\|_t,\quad\forall\mbf{A}\in\mcal{M}_n,
	$$
	where $\mcal{M}_n$ is the collection of all $n\times n$ matrices.

	\newpage
	\section{Technical Proofs}\label{Sec: proof}
	\subsection{Proof of Lemma \ref{funcProdApproxLm}}
	\begin{proof}
		Note that
		$$
		\mbb{E}\left[\frac{\mbf{w}_i^T\mbf{w}_j}{m}\right]=\frac{1}{m}\sum_{k=1}^m\mbb{E}\left[\mbf{w}_{ik}\mbf{w}_{jk}\right]=\sigma_{ij}.
		$$
		Now we consider the Taylor expansion of $\mbf{K}_{ij}(\mbf{U})$ around $\sigma_{ij}$:
		$$
		\mbf{K}_{ij}(\mbf{U})=f(\sigma_{ij})+f'(\sigma_{ij})\left(\frac{\mbf{w}_i^T\mbf{w}_j}{m}-\sigma_{ij}\right)+\frac{1}{2}f''(\eta_{ij})\left(\frac{\mbf{w}_i^T\mbf{w}_j}{m}-\sigma_{ij}\right)^2,
		$$
		where $\eta_{ij}$ is between $\sigma_{ij}$ and $\frac{\mbf{w}_i^T\mbf{w}_j}{m}$. Then the truncation error can be evaluated as follows. For all $\delta>0$,
		\begin{align*}
			\mbb{P}\left(\left|\mbf{K}_{ij}(\mbf{U})-\hat{\mbf{K}}_{ij}(\mbf{U})\right|>\delta\right) & =\mbb{P}\left(\frac{1}{2}|f''(\eta_{ij})|\left(\frac{\mbf{w}_i^T\mbf{w}_j}{m}-\sigma_{ij}\right)^2>\delta\right)\\
			& \leq\mbb{P}\left(\frac{1}{2}M\left(\frac{\mbf{w}_i^T\mbf{w}_j}{m}-\sigma_{ij}\right)^2>\delta\right)\\
			& =\mbb{P}\left(\left|\frac{\mbf{w}_i^T\mbf{w}_j}{m}-\sigma_{ij}\right|>\sqrt{\frac{2\delta}{M}}\right).
		\end{align*}
		So it suffices to evaluate the tail probability $\mbb{P}\left(\left|\frac{\mbf{w}_i^T\mbf{w}_j}{m}-\sigma_{ij}\right|>\sqrt{\frac{2\delta}{M}}\right)$. Note that 
		$$
		\mbf{w}_j|\mbf{w}_i\sim\mcal{N}_m\left(\frac{\sigma_{ij}}{\sigma_{ii}}\mbf{w}_i,\left(\sigma_{jj}-\frac{\sigma_{ij}^2}{\sigma_{ii}}\right)\mbf{I}_m\right),
		$$
		we get
		$$
		\mbf{w}_i^T\mbf{w}_j|\mbf{w}_i\sim\mcal{N}\left(\frac{\sigma_{ij}}{\sigma_{ii}}\mbf{w}_i^T\mbf{w}_i,\left(\sigma_{jj}-\frac{\sigma_{ij}^2}{\sigma_{ii}}\right)\mbf{w}_i^T\mbf{w_i}\right)
		$$
		Therefore, given $\mbf{w}_i$, the random variable $\mbf{w}_i^T\mbf{w}_j$ is sub-Gaussian with sub-Gaussian parameter $\sigma_{ii}^{-1}s_{i,j}\mbf{w}_i^T\mbf{w}_i$, where $s_{ij}=\sigma_{ii}\sigma_{jj}-\sigma_{ij}^2$. Since
		\begin{align*}
			\mbb{P}\left(\left|\frac{\mbf{w}_i^T\mbf{w}_j}{m}-\sigma_{ij}\right|>\sqrt{\frac{2\delta}{M}}\right) & \leq\mbb{P}\left(\left|\frac{\mbf{w}_i^T\mbf{w}_j}{m}-\frac{\sigma_{ij}}{\sigma_{ii}}\frac{\mbf{w}_i^T\mbf{w}_i}{m}\right|>\frac{1}{2}\sqrt{\frac{2\delta}{M}}\right)+\\
			& \hspace{3cm}\mbb{P}\left(\left|\frac{\sigma_{ij}}{\sigma_{ii}}\frac{\mbf{w}_i^T\mbf{w}_i}{m}-\sigma_{ij}\right|>\frac{1}{2}\sqrt{\frac{2\delta}{M}}\right)\\
			& :=(I)+(II),
		\end{align*}
		it suffices to provide bounds for (\textit{I}) and (\textit{II}).
		
		For $(II)$, note that $\mbf{w}_i\sim\mcal{N}_m(\mbf{0},\sigma_{ii}\mbf{I}_m)$, we have $\frac{1}{\sigma_{ii}}\mbf{w}_i^T\mbf{w}_i\sim\chi_m^2$, which implies that
		\begin{align*}
			\mbb{E}\left[e^{\lambda\left(\frac{\mbf{v}_i^T\mbf{v}_i}{\sigma_{ii}}-m\right)}\right] & =e^{-\lambda m}\mbb{E}\left[e^{\lambda\frac{\mbf{v}_i^T\mbf{v}_i}{\sigma_{ii}}}\right]=e^{-\lambda m}(1-2\lambda)^{-\frac{m}{2}},\quad\mrm{for }\lambda<\frac{1}{2}\\
			& =\left(\frac{e^{-\lambda}}{\sqrt{1-2\lambda}}\right)^m\\
			& \leq e^{2m\lambda^2},\quad\mrm{for all }|\lambda|<\frac{1}{4}\\
			& =e^{\frac{4m\lambda^2}{2}},\quad\mrm{for all }|\lambda|<\frac{1}{4},
		\end{align*}
		i.e., $\frac{1}{\sigma_{ii}}\mbf{w}_i^T\mbf{w}_i$ is sub-exponential with parameters $(2\sqrt{m}, 4)$. Hence, for $\sigma_{ij}\neq0$, by Theorem \ref{TailBoundForSubExp} in Appendix A.1, we have
		\begin{align*}
			(II) & =\mbb{P}\left(\left|\frac{1}{\sigma_{ii}}\mbf{w}_i^T\mbf{w}_i-m\right|>\frac{2m}{|\sigma_{ij}|}\sqrt{\frac{2\delta}{M}}\right)\\
			& \leq2\exp\left\{-\left(\frac{\delta m}{\sigma_{ij}^2M}\wedge\frac{m}{4|\sigma_{ij}|}\sqrt{\frac{2\delta}{M}}\right)\right\}\numberthis\label{boundOnII}
		\end{align*}
		If $\sigma_{ij}=0$, then $(II)=0$.
		
		For $(I)$, by Hoeffding inequality, we have
		\begin{align*}
			(I) & =\mbb{E}_{\mbf{w}_i}\left[\mbb{P}\left(\left.\left|\frac{\mbf{w}_i^T\mbf{w}_j}{m}-\frac{\sigma_{ij}}{\sigma_{ii}}\frac{\mbf{w}_i^T\mbf{w}_i}{m}\right|>\frac{1}{2}\sqrt{\frac{2\delta}{M}}\right|\mbf{w}_i\right)\right]\\
			& =\mbb{E}_{\mbf{w}_i}\left[\mbb{P}\left(\left.\left|\mbf{w}_i^T\mbf{w}_j-\frac{\sigma_{ij}}{\sigma_{ii}}\mbf{w}_i^T\mbf{w}_i\right|>\frac{m}{2}\sqrt{\frac{2\delta}{M}}\right|\mbf{w}_i\right)\right]\\
			& \leq\mbb{E}_{\mbf{w}_i}\left[2\exp\left\{-\frac{\sigma_{ii}m^2\delta}{4Ms_{ij}\mbf{w}_i^T\mbf{w}_i}\right\}\right]\numberthis\label{IHoeffding}
		\end{align*}
		From Theorem A in \citet{inglot2010inequalities} stated that for a random variable $\chi\sim\chi_m^2$, the $100(1-\alpha)$th percentile is upper bounded by $m+\log(1/\alpha)+2\sqrt{m\log(1/\alpha)}$, which is of the order $\mcal{O}(m)$ as $m\to\infty$. Now, since $\sigma_{ii}^{-1}\mbf{w}_i^T\mbf{w}_i\sim\chi_m^2$, we get for any $\alpha\in(0,1)$,
		$$
		\mbb{P}\left(\sigma_{ii}^{-1}\mbf{w}_i^T\mbf{w}_i\geq m+2\log\frac{1}{\alpha}+2\sqrt{2m\log\frac{1}{\alpha}}\right)=\alpha.
		$$
		Let $q(\alpha,m)=m+2\log(1/\alpha)+2\sqrt{2m\log(1/\alpha)}$. Since the function $\exp\{-a/x\}$ is increasing in $x$ for $a>0$, we can further bound (\ref{IHoeffding}) as follows:
		\begin{align*}
			(I) & \leq\mbb{E}_{\mbf{w}_i}\left[2\exp\left\{-\frac{m^2\delta}{4Ms_{ij}\sigma_{ii}^{-1}\mbf{w}_i^T\mbf{w}_i}\right\}\mbb{I}_{\{\sigma_{ii}^{-1}\mbf{w}_i^T\mbf{w}_i\leq q(\alpha,m)\}}\right]+\\
			& \hspace{4cm}\mbb{E}_{\mbf{w}_i}\left[2\exp\left\{-\frac{m^2\delta}{4Ms_{ij}\sigma_{ii}^{-1}\mbf{w}_i^T\mbf{w}_i}\right\}\mbb{I}_{\{\sigma_{ii}^{-1}\mbf{w}_i^T\mbf{w}_i\geq q(\alpha,m)\}}\right]\\
			& \leq2\exp\left\{-\frac{m^2\delta}{4Ms_{ij}q(\alpha,m)}\right\}+2\mbb{P}\left(\sigma_{ii}^{-1}\mbf{w}_i^T\mbf{w}_i\geq q(\alpha,m)\right)\\
			& \leq2\exp\left\{-\frac{m^2\delta}{4Ms_{ij}q(\alpha,m)}\right\}+2\alpha.
		\end{align*}
		By choosing $\alpha=\exp\{-m\}$, we get $q(\alpha,m)=m+2m+2\sqrt{m^2}=5m$ so that
		\begin{align*}
			(I) & \leq2\exp\left\{-\frac{m^2\delta}{20Mms_{ij}}\right\}+2\exp\{-m\}\\
			& \leq2\exp\left\{-m\left(1\wedge\frac{\delta}{20Ms_{ij}}\right)\right\}.\numberthis\label{boundOnI}
		\end{align*}
		Combining (\ref{boundOnI}) and (\ref{boundOnII}), we obtain for all $\delta>0$,
		\begin{align*}
			\mbb{P}\left(\left|\mbf{K}_{ij}(\mbf{U})-\hat{\mbf{K}}_{ij}(\mbf{U})\right|>\delta\right) & \leq (I)+(II)\\
			& \leq4\exp\left\{-m\left(1\wedge\frac{\delta}{20Ms_{ij}}\wedge\frac{\delta}{M\sigma_{ij}^2}\wedge\frac{
				1}{4|\sigma_{ij}|}\sqrt{\frac{2\delta}{M}}\right)\right\}.
		\end{align*}
	\end{proof}
	
	\begin{remark}\label{funcProdApproxRmk}
		The condition (\ref{asBoundedAssumption}) can be weakened as follows:
		$$
		f''\left(\lambda\sigma_{ij}+(1-\lambda)\frac{\mbf{w}_i^T\mbf{w}_j}{m}\right)=\mcal{O}_p(1)
		$$
		for all $\lambda\in[0,1]$. In such case, the evaluation of truncation error can be modified as follows: For all $\delta>0$, there exists $M_\delta>0$ such that
		$$
		\mbb{P}\left(\left|f''(\eta_{ij})\right|>M_\delta\right)<\frac{\delta}{2},
		$$
		where $\eta_{ij}=\lambda\sigma_{ij}+(1-\lambda)\frac{\mbf{w}_i^T\mbf{w}_j}{m}$ for some $\lambda\in[0,1]$ and then
		\begin{align*}
			\mbb{P}\left(\left|\mbf{K}_{ij}(\mbf{U})-\hat{\mbf{K}}_{ij}(\mbf{U})\right|>\delta\right) & =\mbb{P}\left(\frac{1}{2}f''(\eta_{ij})\left(\frac{\mbf{w}_i^T\mbf{w}_j}{m}-\sigma_{ij}\right)^2>\delta\right)\\
			& \leq\mbb{P}\left(\left\{\frac{1}{2}f''(\eta_{ij})\left(\frac{\mbf{w}_i^T\mbf{w}_j}{m}-\sigma_{ij}\right)^2>\delta\right\}\cap\left\{|f''(\eta_{ij})|\leq M_\delta\right\}\right)+\\
			& \hspace{5cm}\mbb{P}\left(|f''(\eta_{ij})|>M_\delta\right)\\
			& \leq\mbb{P}\left(\left|\frac{\mbf{w}_i^T\mbf{w}_j}{m}-\sigma_{ij}\right|>\sqrt{\frac{2\delta}{M_\delta}}\right)+\frac{\delta}{2}\\
			& \leq4\exp\left\{-m\left(1\wedge\frac{\delta}{20M_\delta s_{ij}}\wedge\frac{\delta}{M_\delta\sigma_{ij}^2}\wedge\frac{1}{4|\sigma_{ij}|}\sqrt{\frac{2\delta}{M_\delta}}\right)\right\}+\frac{\delta}{2}
		\end{align*}
		Now, we can choose $m>\left(1\wedge\frac{\delta}{20M_\delta s_{ij}}\wedge\frac{\delta}{M_\delta\sigma_{ij}^2}\wedge\frac{1}{4|\sigma_{ij}|}\sqrt{\frac{2\delta}{M_\delta}}\right)^{-1}\log\frac{8}{\delta}$ so that
		$$
		\mbb{P}\left(\left|\mbf{K}_{ij}(\mbf{U})-\hat{\mbf{K}}_{ij}(\mbf{U})\right|>\delta\right)<\frac{\delta}{2}+\frac{\delta}{2}=\delta
		$$
		and hence $\mbf{K}_{ij}(\mbf{U})=\hat{\mbf{K}}_{ij}(\mbf{U})+o_p(1)$.
	\end{remark}

	\subsection{Proof of Lemma \ref{PredErrApprox}}
	\begin{proof}
		\begin{enumerate}[(i)]
			\item When $f(x)=x$, we have $\mbf{K}(\mbf{U})=\frac{1}{m}\mbf{UU}^T$ and since the hidden random vectors $\mbf{u}_1,\ldots,\mbf{u}_m$ are i.i.d, we can know that
			$$
			\mbf{u}_1\mbf{u}_1^T,\ldots,\mbf{u}_m\mbf{u}_m^T\sim\mrm{ i.i.d. }\mcal{W}_n\left(1,\sum_{l=1}^L\xi_l\mbf{K}_l(\mbf{X})\right),
			$$
			where $\mcal{W}_n\left(1,\sum_{l=1}^L\xi_l\mbf{K}_l(\mbf{X})\right)$ stands for a Wishart distribution with degrees of freedom 1 and covariance matrix $\sum_{l=1}^L\xi_l\mbf{K}_l(\mbf{X})$. Therefore, the Strong Law of Large Numbers implies that as $m\to\infty$,
			$$
			\mbf{K}(\mbf{U})=\frac{1}{m}\mbf{U}\mbf{U}^T=\frac{1}{m}\sum_{i=1}^m\mbf{u}_i\mbf{u}_i^T\to\mbb{E}\left[\mbf{u}_1\mbf{u}_1^T\right]=\sum_{l=1}^L\xi_l\mbf{K}_l(\mbf{X}),\quad\mrm{a.s.}.
			$$
			Since the the map $\psi:\mbf{A}\mapsto\mbf{A}^{-1}$ for non-singular matrix $\mbf{A}$ is continuous, then we can obtain the following result by using the Continuous Mapping Theorem.
			$$
			\left(\tilde{\tau}\mbf{K}(\mbf{U})+\mbf{I}_n\right)^{-1}\to\left(\sum_{l=1}^L\tilde{\tau}\xi\mbf{K}_l(\mbf{X})+\mbf{I}_n\right)^{-1},\quad\mrm{a.s., as }m\to\infty
			$$
			Let $\tilde{\mbf{A}}=(\tilde{\tau}\mbf{K}(\mbf{U})+\mbf{I}_n)^{-1}$, we have
			\begin{align*}
				\max_{1\leq i,j\leq n}|\tilde{\mbf{A}}_{ij}|\leq\|\tilde{\mbf{A}}\|_\infty & \lesssim\|\tilde{\mbf{A}}\|_{\mrm{op}}\\
				& =\lambda_{\max}\left(\left(\tilde{\tau}\mbf{K}(\mbf{U})+\mbf{I}_n\right)^{-1}\right)\\
				& =\frac{1}{\tilde{\tau}\lambda_{\min}(\mbf{K}(\mbf{U}))+1}\leq1<\infty.
			\end{align*}
			Therefore, by Bounded Convergence Theorem,
			$$
			\mbf{A}:=\mbb{E}\left[\left(\tilde{\tau}\mbf{K}(\mbf{U})+\mbf{I}_n\right)^{-1}\right]\to\left(\sum_{l=1}^L\tilde{\tau}\xi\mbf{K}_l(\mbf{X})+\mbf{I}_n\right)^{-1},\quad\mrm{a.s. as }m\to\infty.
			$$
			Asymptotically, we get as $m\to\infty$.
			$$
			R\simeq\mbf{y}^T\left(\sum_{l=1}^L\tilde{\tau}\xi\mbf{K}_l(\mbf{X})+\mbf{I}_n\right)^{-2}\mbf{y}.
			$$
			
			\item Note that equation (\ref{KmatApprox}) can be further written as
			$$
			\mbf{K}(\mbf{U})=f[\mbf{\Sigma}]+o_P(1),
			$$
			or equivalently, $\mbf{K}(\mbf{U})\xrightarrow{P}f[\mbf{\Sigma}]$ as $m\to\infty$ element-wisely. Similarly, under the assumption of $\|\mbf{K}(\mbf{U})\|_{\mrm{op}}<\infty$ a.s., we have
			$$
			\mbb{E}\left[\tilde{\tau}\mbf{K}(\mbf{U})+\mbf{I}_n\right]\to\tilde{\tau}f[\mbf{\Sigma}]+\mbf{I}_n.
			$$
			Hence by Bounded Convergence Theorem and Continuous Mapping Theorem, we have
			$$
			\mbf{A}=\mbb{E}\left[\left(\tilde{\tau}\mbf{K}(\mbf{U})+\mbf{I}_n\right)^{-1}\right]\to\left(\tilde{\tau}f[\mbf{\Sigma}]+\mbf{I}_n\right)^{-1},\quad\mrm{as }m\to\infty,
			$$
			which shows that as $m\to\infty$, 
			$$
			R\simeq\mbf{y}^T\left(\tilde{\tau}f\left[\sum_{l=1}^L\xi\mbf{K}_l(\mbf{X})\right]+\mbf{I}_n\right)^{-2}\mbf{y}.
			$$
		\end{enumerate}
	\end{proof}

	\subsection{Proof of Proposition \ref{propKNN1}}
	\begin{proof}
		The result follows by noting that
		\begin{align*}
			APELMM & =\mbb{E}\left[(\mbf{y}-\mbb{E}[\mbf{a}|\mbf{y}])^T(\mbf{y}-\mbb{E}[\mbf{a}|\mbf{y}])\right]\\
			& =\mbb{E}\left[\mbf{y}^T\left(\mbf{I}_n-\tilde{\sigma}_R^2\mbf{\Sigma}(\tilde{\sigma}_R^2\mbf{\Sigma}+\mbf{I}_n)^{-1}\right)^T\left(\mbf{I}_n-\tilde{\sigma}_R^2\mbf{\Sigma}(\tilde{\sigma}_R^2\mbf{\Sigma}+\mbf{I}_n)^{-1}\right)\mbf{y}\right]\\
			& =\mbb{E}\left[\mbf{y}^T\left((\tilde{\sigma}_R^2\mbf{\Sigma}+\mbf{I}_n)^{-1}\right)^2\mbf{y}\right]\\
			& =\mrm{tr}\left[\left((\tilde{\sigma}_R^2\mbf{\Sigma}+\mbf{I}_n)^{-1}\right)^2\left(\sigma_R^2\mbf{\Sigma}+\phi\mbf{I}_n\right)\right]\\
			& =\phi\mrm{tr}\left[(\tilde{\sigma}_R^2\mbf{\Sigma}+\mbf{I}_n)^{-1}\right]\\
			& =\phi\sum_{i=1}^n\left(\tilde{\sigma}_R^2\lambda_i(\mbf{\Sigma})+1\right)^{-1}.
		\end{align*}
	\end{proof}

	\newpage
	\section{More Simulation Results}
	Figure \ref{nonlinear2} demonstrates the performance of LMM and KNN in terms of prediction error when the inverse logistic function and the polynomial function of order 2 are used.
	
	\begin{figure}[htbp]
		\centering
		\includegraphics[width=0.45\textwidth]{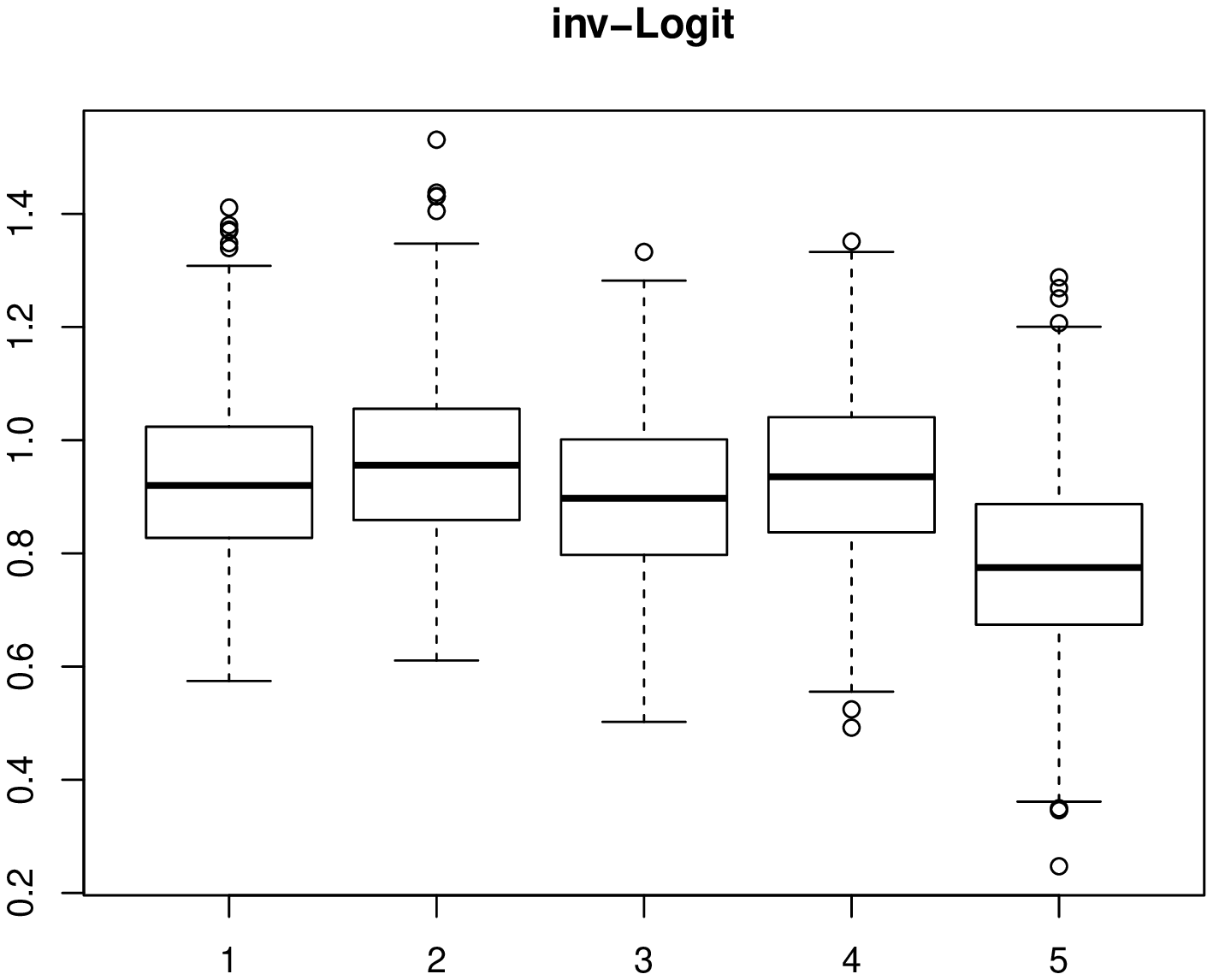}
		\includegraphics[width=0.45\textwidth]{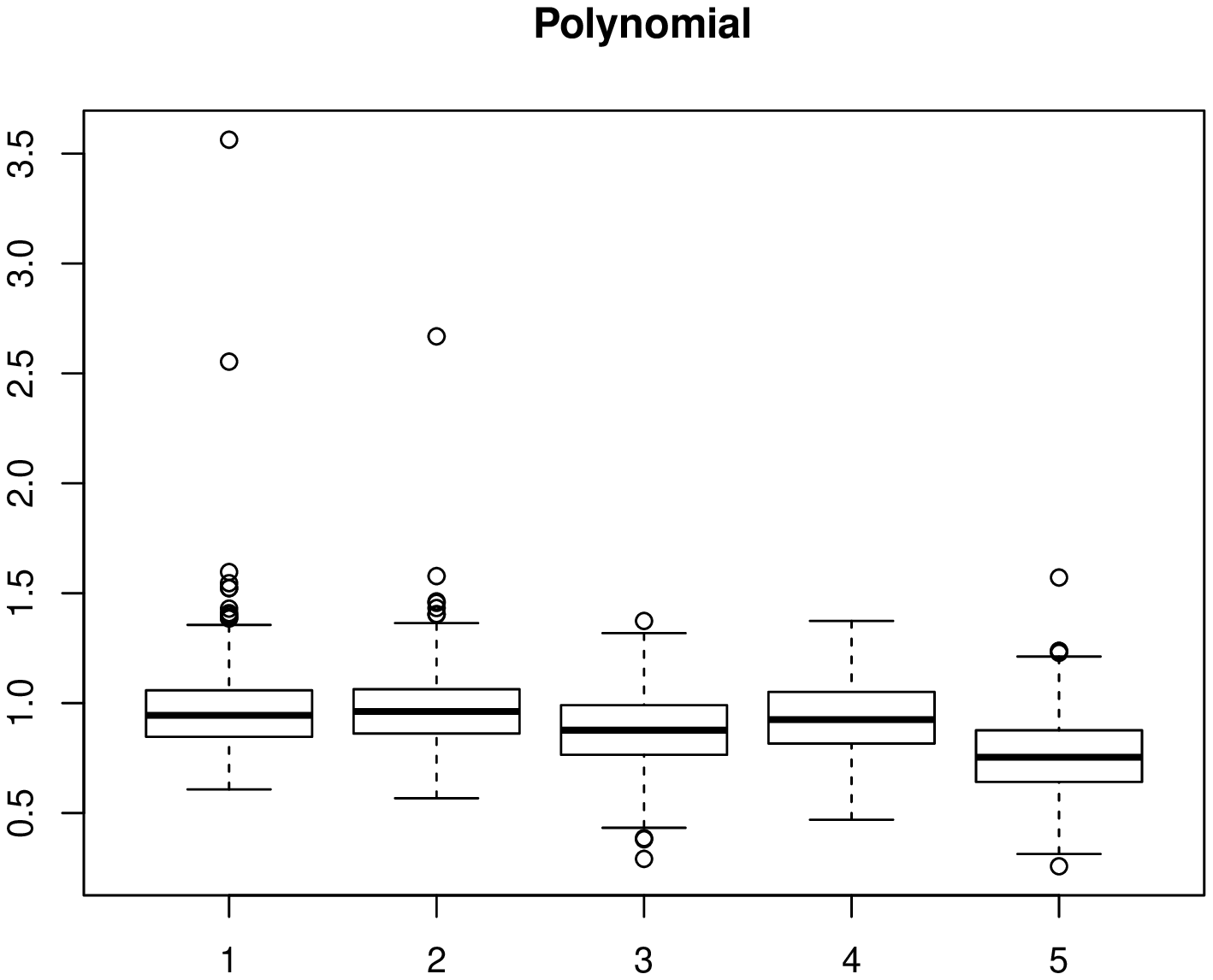}
		\caption{The boxplots for linear mixed models (LMM) and kernel neural network (KNN) in terms of prediction errors. The left panel shows the results when an inverse logistic function is used and the right panel shows the results when a polynomial function of order 2 is used. In the horizontal axis, ``1" corresponds to the LMM; ``2" corresponds to the KNN with product input kernel and product output kernel; ``3" corresponds to the KNN with product input and polynomial output; ``4" corresponds to the KNN with polynomial input and product output and ``5" corresponds to the polynomial input and polynomial output.}\label{nonlinear2}
	\end{figure}
\end{appendices}

%
%
%
%
%

\bibliography{KNN}
\end{document}